\documentclass[journal]{IEEEtran}
\IEEEoverridecommandlockouts
\usepackage{graphicx}
\usepackage{amsmath,amsthm,amsfonts,amssymb}
\usepackage{enumerate}
\usepackage{cite}
\usepackage{bm}
\usepackage{bbm}
\usepackage{url}
\usepackage{array}
\usepackage{color,soul}
\usepackage{subcaption} 
\usepackage{enumitem}
\usepackage{romannum}
\usepackage{amssymb}

\usepackage{multirow}
\usepackage{booktabs}

\theoremstyle{plain}
\newtheorem{thm}{Theorem}

\newtheorem{prop}[thm]{Proposition}

\newtheorem{rem}{Remark}

\allowdisplaybreaks[4]

\begin{document}
\title{Accelerating Heterogeneous Agent Collaboration in Dynamic Edge Networks}

\author{
Tianji~He,
Yulin~Shao,
Fen~Hou

\thanks{T. He and F. Hou are with the State Key Laboratory of Internet of Things for Smart City, University of Macau, Macau, China (e-mails: tianjihe296@gmail.com, fenhou@um.edu.mo). T. He is also with the Department of Electrical and Computer Engineering, The University of Hong Kong, China.}
\thanks{Y. Shao is with the Department of Electrical and Computer Engineering, The University of Hong Kong, Hong Kong, China (e-mail: ylshao@hku.hk).}
}

\maketitle
\thispagestyle{empty}

\begin{abstract}
Deploying large language models (LLMs) at the network edge is hindered by their enormous cost, yet the reasoning quality they provide remains indispensable. Heterogeneous collaboration between edge small models and a server LLM has emerged as a promising direction, but existing methods fail under the dynamic conditions of multi-user contention, autoregressive generation, and time-varying resources.
This paper puts forward a process reward model (PRM)-aided two-stage decoupled acceleration (PRADA) framework, which is built on a fundamental change of perspective: instead of querying a PRM online, which cripples multi-user systems with prohibitive latency, we use the PRM solely as an offline teacher. Its reasoning-quality intuition is fully distilled into a lightweight policy that screen each step locally, without any context upload, while a Lagrangian scheduler at the server resolves resource contention through a threshold-structured policy.
Across diverse reasoning benchmarks, PRADA retains the vast majority of the LLM's accuracy while substantially reducing end-to-end latency. 
The results further reveal threshold effects for both server parallel capacity and total bandwidth: performance saturates beyond critical resource levels, after which the system bottleneck shifts from queuing to computation or from communication to contention. These structural findings provide actionable guidance for joint provisioning of computation and communication resources without requiring per-benchmark tuning.
\end{abstract}

\begin{IEEEkeywords}
Heterogeneous agent collaboration, dynamic edge network, process reward model, decoupled scheduling.
\end{IEEEkeywords}

\section{Introduction}
\subsection{Background}\label{sec:background}
The rapid evolution of large language models (LLMs) has unlocked unprecedented reasoning capabilities~\cite{brown2020language,wei2022chain,lewis2020retrieval,schick2023toolformer}, yet deploying them pervasively at the edge remains a formidable challenge due to their enormous computational and communication demands \cite{qu2025mobile,shao2021federated}. Simultaneously, modern edge devices are increasingly equipped with capable, albeit smaller, language models (SLMs) that can perform inference locally with very low latency~\cite{kang2023knowledge,liu2024ddk,zhou2024jiuzhang3,lu2025demystifying}. This naturally motivates a heterogeneous agent collaboration paradigm: a central server hosting a powerful LLM cooperates with multiple edge users, each carrying a local SLM, to process complex reasoning tasks. By intelligently deciding which generation steps should be handled locally and which should be offloaded to the server, such a collaborative system can, in principle, approach the reasoning quality of the LLM while maintaining much of the speed and privacy advantage of edge-only execution.

Despite this promise, realizing efficient collaboration in realistic \emph{dynamic} edge environments raises three fundamental obstacles that existing approaches largely overlook.
\begin{itemize}[leftmargin=0.5cm]
    \item \textit{Dynamic, multi-step task nature}: Reasoning tasks unfold through autoregressive chain-of-thought generation, where each step produces a segment of the final answer and the accumulated context grows over time. The decision made at any single step not only impacts the immediate quality and latency but also changes the cost of subsequent decisions.
    \item \textit{Multi-user resource coupling}: In a dynamic network with stochastically arriving users, all active devices compete for the same finite resources, such as the server's maximum parallel processing units and the total communication bandwidth. An offloading decision for one user inevitably alters queuing states and bandwidth availability for all others, making the scheduling problem tightly coupled across tasks and time.
    \item \textit{The quality-latency trade-off}: Using the server LLM improves reasoning accuracy but incurs a triple delay penalty: communication delay to upload the ever-growing context, queuing delay due to limited server concurrency, and the server's own computation delay. Conversely, executing a step on the local SLM avoids these penalties but risks lower-quality outputs. Balancing this trade-off in a time-varying environment under strict resource constraints is the core challenge.
\end{itemize}

\subsection{Related Work}
Existing studies on collaboration between small and large models largely focus on single-user or single-request settings, where the primary goal is to improve inference efficiency while preserving the reasoning quality of the large model. Consequently, most methods are designed for local collaboration within a single reasoning chain and do not address dynamic multi-user systems with shared and limited server resources. These methods can be broadly categorized by their decision granularity.

\subsubsection{Problem-level routing}
The coarsest category makes a single routing decision for an entire problem. RouteLLM~\cite{ong2025routellm} learns a router to choose between strong and weak models, while FrugalGPT~\cite{chen2024frugalgpt} implements an adaptive cascade of LLMs to reduce cost. Such approaches are simple and easy to deploy, but because they commit to one model for the whole task, they cannot exploit step-varying difficulty or switch models mid-generation, leading to a rigid accuracy-efficiency trade-off on complex reasoning tasks.

\subsubsection{Step- and token-level collaboration}
To introduce finer control, later works operate at the level of individual reasoning steps or even tokens. Early step-level methods did not use an explicit process reward model (PRM); instead, they relied on local signals such as token confidence, entropy, or immediate consistency checks to decide when to invoke a stronger model. While better aligned with the internal structure of reasoning than problem-level routing, these approaches often suffer from myopic decisions because they lack a forward-looking estimate of how a partial reasoning state will affect final outcomes. Token-level methods, exemplified by speculative decoding schemes such as HSL~\cite{DBLP:conf/edgefm/HaoJJ0C24}, HLM~\cite{11140540}, and CITER~\cite{zheng2025citer}, push granularity even further by switching models on a per-token basis. They can achieve substantial speedups, but their decision mechanisms remain driven by local uncertainty or verification signals and do not consider the global resource context or long-horizon quality consequences.

\subsubsection{PRM-guided step-level methods}
More recent step-level methods incorporate a PRM to overcome the shortsightedness of local heuristics. A PRM evaluates not only the current partial output but also its potential to lead to a correct final answer, thus providing a more forward-looking training signal. Representative works include RSD~\cite{liao2025rewardguided}, which uses PRM-based rewards to guide speculative reasoning, and G-Boost~\cite{fan2025g}, which leverages PRM signals to steer search over collaboration paths. These methods better capture long-range reasoning quality and, in principle, produce more informed collaboration decisions. However, this advantage is accompanied by a critical drawback in resource-constrained multi-user settings: the PRM itself is typically as large as the LLM, and invoking it for every step of every user introduces prohibitive online computation and memory overhead, exactly the opposite of what acceleration aims to achieve.

This observation directly motivates a key design principle of our work. Instead of deploying the PRM as an online component, we use it exclusively during offline training to provide dense reward supervision and distill its forward-looking evaluation capability into an extremely lightweight decision network. At deployment time, each edge SLM makes a binary preliminary decision, i.e., whether the current step should remain local or be submitted to the server scheduler, using a network with only a few hundred thousand parameters, thereby completely eliminating PRM inference latency from the online loop.

\subsubsection{Multi-agent homogeneous collaboration}
Besides collaboration between small and large models, multi-agent collaboration has also become an important direction for improving performance on complex tasks \cite{cui2024llmind,luo2026cayley}. These methods usually enhance output quality through division of labor, discussion, mutual critique, voting, or role-based cooperation among multiple agents~\cite{lowe2017multi,tampuu2017multiagent,shao2024theory}. However, they typically assume agents with relatively comparable roles and focus on interaction protocols \cite{qian2024scaling}, rather than on capability asymmetry and resource-constrained scheduling. By contrast, the setting considered in this paper is highly asymmetric: the server-side LLM offers stronger reasoning quality but incurs higher latency and cost, whereas the edge-side SLM is faster but less capable. Therefore, the problem studied here is not collaboration among homogeneous peers, but efficient collaboration and scheduling among heterogeneous agents in a dynamic multi-user edge network.

\subsection{Our Contributions}
While prior works have demonstrated the importance of efficient collaboration between small and large models, they are largely designed around a single reasoning chain and a static resource assumption. Once the system is placed in a practical dynamic edge inference system, none of them can handle the combination of challenges highlighted in Section \ref{sec:background}. 
To fill this gap, this paper puts forward a PRM-aided two-stage decoupled acceleration (PRADA) framework. Our main contributions are summarized as follows:

\begin{itemize}[leftmargin=0.5cm]
\item We provide the first principled formulation of heterogeneous agent collaboration in dynamic edge networks as a constrained sequential decision problem. Our formulation explicitly captures stochastic user arrivals, per-step chain-of-thought generation with growing context, and coupled competition for limited server concurrency and communication bandwidth. Crucially, we build a unified latency model that expresses all sources of delay, inclduing computation, communication, and queuing, in a single, real-time metric. The computation component is grounded in the actual LLM inference architecture: it separately accounts for the prefill and decode phases using FLOPs-level characterization, precisely capturing the strong dependence of both the server-side and local inference time on context length and the number of generated tokens. This rigorous modeling avoids the coarse heuristics prevalent in prior work and enables consistent optimization of a composite objective that trades off reasoning quality against total end-to-end latency through a single, interpretable parameter.

\item Building on the formulation, we design the PRADA framework, which rests on two essential and complementary ideas. The first is a two-stage architecture that decouples the global decision problem. At the edge, a compact screening network makes per-user, per-step binary nominations without uploading any context; only the nominated candidates are forwarded to the server, reducing the candidate set by an order of magnitude and eliminating unnecessary communication. At the server, a Lagrangian-based scheduler resolves the final actions under real-time resource constraints, yielding a threshold-structured scheduling policy and a closed-form rule for bandwidth allocation. The second idea is to use the PRM exclusively as an offline teacher. Instead of enduring its prohibitive online inference cost, we distill its forward-looking quality assessment into the lightweight screening policy during training. The reward model is never invoked at deployment time, yet the resulting policy preserves its ability to judge whether a step is likely to benefit from stronger reasoning. We further prove that, under mild conditions, a first-stage decision to stay local remains optimal even after communication and queuing penalties are added back, guaranteeing that valuable offloading candidates are never prematurely discarded.

\item Extensive simulations across several reasoning benchmarks reveal structural properties of the system. We show that PRADA consistently retains most of the accuracy gains achievable by always using the large model while substantially reducing end-to-end delay. More importantly, we uncover clear threshold effects for both the server parallel processing capacity and the total communication bandwidth. As either resource increases, accuracy and latency improve rapidly up to a critical point, after which the gains saturate and the system bottleneck shifts from one resource to another, for example, from queuing to computation, or from communication to server contention. These findings provide concrete, actionable guidance for jointly provisioning computation and communication resources to achieve a desired operating point on the quality-latency frontier, without requiring per-benchmark tuning or detailed prior knowledge of the task mix.
\end{itemize}

\section{System Model}
\label{sec:system_model}

We consider a mobile edge network consisting of a central server and multiple users, organized in a star topology as illustrated in Figure~\ref{fig:system_model}. The server hosts a powerful LLM, while each user is equipped with a local SLM. These two types of models possess fundamentally different capabilities and costs: the SLM is lightweight and fast but has limited reasoning ability, whereas the LLM provides stronger reasoning quality at the cost of higher computational latency. In this sense, the server LLM and edge SLMs form a team of heterogeneous agents, where heterogeneity manifests not only in their reasoning performance but also in their execution latency, resource consumption, and communication overhead.

\begin{figure}[t]
    \centering
    \includegraphics[width=0.5\textwidth]{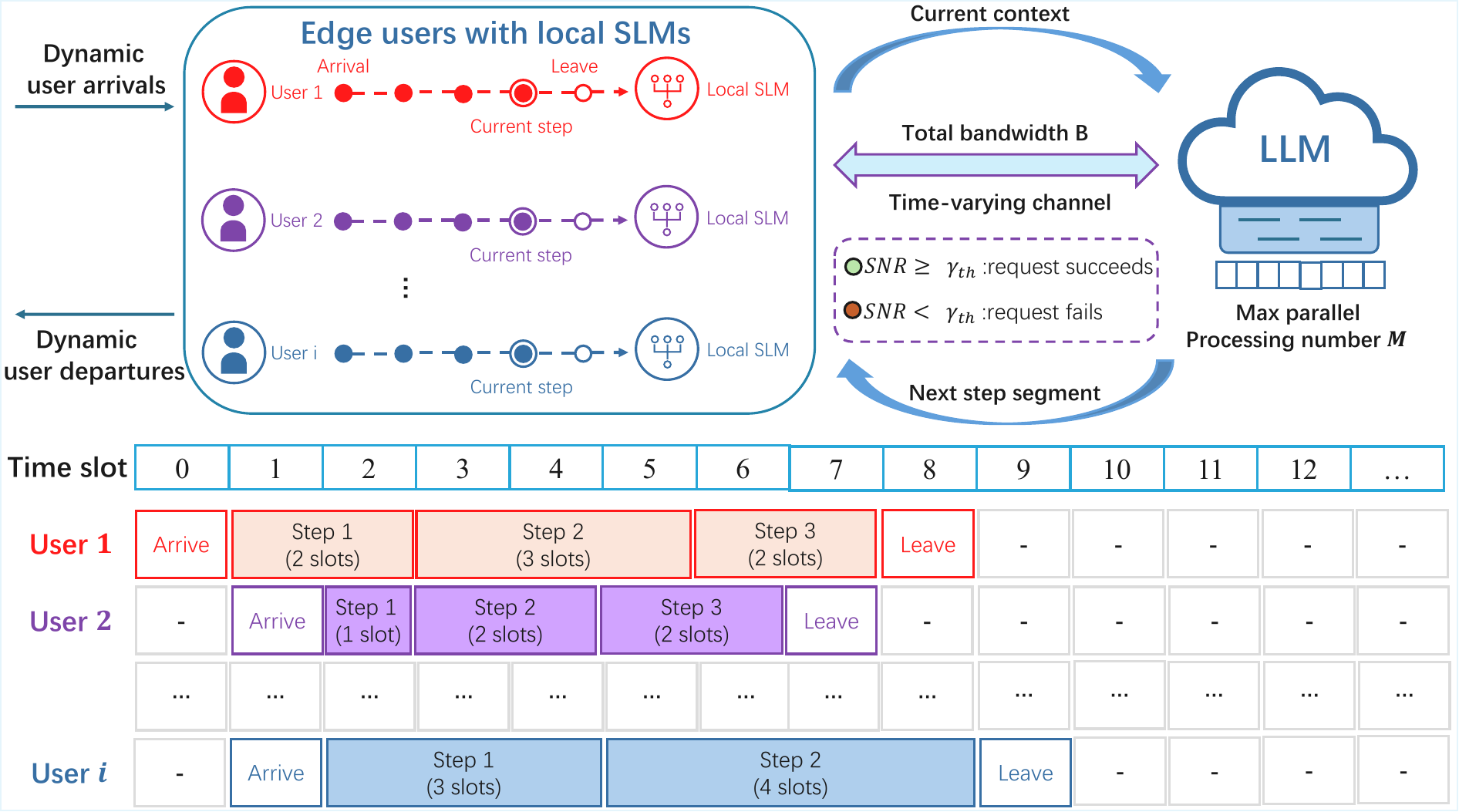}
    \caption{System model of dynamic heterogeneous agent collaboration in edge networks.}
    \label{fig:system_model}
\end{figure}

Users arrive dynamically over time. Each user comes with a reasoning task (e.g., planning a travel itinerary, solving a math problem, or answering a complex question) and leaves the network once its task is completed. The total communication bandwidth between users and the server is limited to $B$. The server can process at most $M$ requests in parallel, forming a shared resource bottleneck. The dynamic nature of user arrivals, together with the heterogeneity of the two agent types, makes the system inherently time-varying and poses a core challenge for efficient collaboration.

In this dynamic edge network, a naive approach would forward every user request directly to the server LLM. However, this quickly becomes infeasible because of the limited server concurrency and communication bandwidth. In addition, users join and leave stochastically,
and the wireless channel SNR varies over time, making static scheduling suboptimal. Thus, accelerating collaboration among heterogeneous agents requires a mechanism that adaptively decides, at each generation step, whether to use the local SLM (fast but weaker) or to offload to the server LLM (stronger but slower), while explicitly accounting for resource contention among multiple users. The goal is to achieve high reasoning accuracy with low overall latency.

\subsection{System Slots and Reasoning Steps}

To capture the system dynamics, we introduce two time scales: a fine-grained system slot for scheduling decisions and a coarser-grained reasoning step for task progression. This distinction is crucial for accurately modeling both the multi-user resource contention and the internal state evolution of individual tasks.

\subsubsection{System slots}
We divide time into discrete slots indexed by $t = 0,1,2,\dots$, and model the number of new users arriving in slot $t$ (i.e., the number of new tasks initiated in slot $t$) by a Poisson process:
\begin{equation}
    n_t \sim \mathrm{Poisson}(\lambda),
    \label{eq:poisson_arrival}
\end{equation}
where $\lambda$ is the average arrival rate. Each user enters the system together with its task request and remains until the task is completed. 

Because new users may arrive at every slot, the system must perform decision-making at the slot level. Specifically, the system must determine, at each slot, whether a request should remain at the local SLM or be served by the edge LLM. 
As users enter and depart, the set of active tasks evolves over time, causing the server workload and communication demand to fluctuate. This is a realistic feature that any practical scheduler must explicitly handle.

\subsubsection{Reasoning steps}
For each user, a reasoning task is not processed in a single monolithic pass. Consistent with the autoregressive nature of chain-of-thought prompting, each task unfolds progressively over a sequence of discrete reasoning steps, indexed by $k = 0,1,2,\dots$. Each step corresponds to the generation of a coherent segment of the response, such as a sentence or a reasoning fragment. The execution of a single reasoning step takes a non-zero amount of time and may therefore span multiple system slots. The duration of a step depends on the chosen execution target (local SLM or server LLM) and the current context length.

For each user/task $i$, we denote its initial query by $q_i$. The task proceeds through a series of reasoning steps and terminates when an end-of-sequence token is generated or a maximum step limit is reached. Let $y_{i}^{(k)}$ be the content generated during the $k$-th step. The {context} available at the beginning of step $k$ is the concatenation of the initial query and all previously generated content:
\begin{equation}
x_{i}^{(k)} =
\begin{cases}
q_i, & k = 0, \\
q_i \oplus y_{i}^{(0)} \oplus y_{i}^{(1)} \oplus \cdots \oplus y_{i}^{(k-1)}, & k \geq 1,
\end{cases}
\label{eq:context_def}
\end{equation}
where $\oplus$ denotes string concatenation. As generation proceeds, the context $x_{i}^{(k)}$ grows in length, which not only influences the difficulty of the current reasoning step but also increases the volume of data that must be transmitted if the step is offloaded to the server.

\subsubsection{Interleaving of the two scales}
The two time scales are interleaved as follows. 
\begin{itemize}[leftmargin=0.5cm]
    \item At any system slot $t$, a user is said to be \textit{active} if it is ready to initiate a new reasoning step. An active user competes for system resources (e.g., bandwidth, server admission) and requires a scheduling decision.
    \item Once a decision is made and the step begins execution, the user becomes \textit{inactive} for the duration of that step's execution. During this inactive period, the task's internal context remains static, and it does not participate in any resource contention.
    \item Upon completion of the current step, the task transitions back to the active state, its context is updated with the newly generated content, and it becomes eligible for scheduling again in the next system slot.
\end{itemize}   

Consequently, the set of users competing for resources in slot $t$, denoted by $\mathcal{U}_t$, is composed of two distinct groups:
\begin{equation}
\mathcal{U}_t = \mathcal{U}_t^{\text{new}} \cup \mathcal{U}_t^{\text{comp}},
\label{eq:active_users}
\end{equation}
where $\mathcal{U}_t^{\text{new}}$ is the set of users newly arrived in slot $t$, and $\mathcal{U}_t^{\text{comp}}$ is the set of existing users whose previous reasoning step was completed just before slot $t$ begins. This formulation explicitly captures the fact that the set of contending users evolves dynamically, driven by both exogenous arrivals and endogenous task progress.

\begin{rem}[An illustrative example]
    Figure~\ref{fig:slot_step} provides a concrete illustration of the interplay between slots and steps for a single task. In this example, the task arrives in slot $t=1$, becomes active, and is scheduled. Its first step ($k=0$) executes over slots $t=1$ and $t=2$, during which the task is inactive. The step completes at the end of slot $t=2$. The task then becomes active again at slot $t=3$ for its second step ($k=1$), and the cycle repeats.
\end{rem}

\begin{figure}[t]
    \centering
    \includegraphics[width=0.5\textwidth]{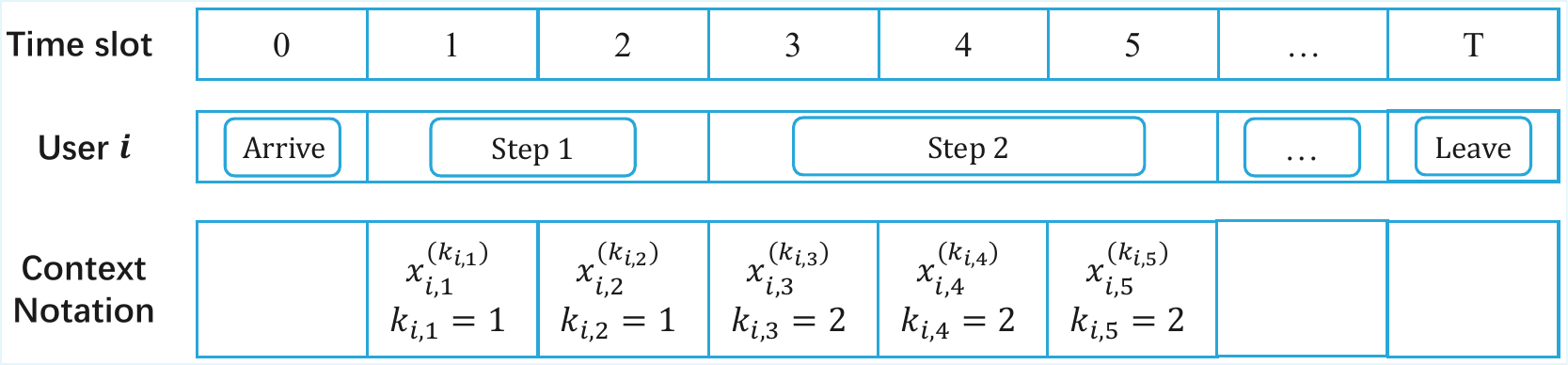}
    \caption{Illustration of the interplay between slots and steps.}
    \label{fig:slot_step}
\end{figure}

To unambiguously refer to the context of an active task $i$ at slot $t$, we adopt the following notation.
\begin{itemize}[leftmargin=0.5cm]
    \item Let $k_{i,t}$ denote the reasoning-step index that task $i$ is about to execute when it becomes active in slot $t$.
    \item The context associated with this decision epoch is then denoted by $x_{i,t}^{(k_{i,t})}$.
    \item For brevity, when the step index is clear from context, we may simply write $x_{i,t}$ with the understanding that it refers to the appropriate $x_{i,t}^{(k_{i,t})}$.
\end{itemize} 

During the inactive periods, the task does not require a context reference in our scheduling formulation, so no ambiguity arises.

\subsection{Latency Decomposition}
For an active task, the total latency incurred to complete its upcoming reasoning step depends on the chosen execution path. We decompose this latency into three major components: communication delay, queuing delay, and computational delay.

\subsubsection{Communication latency} 
When a step is offloaded to the server, the current context $x_{i,t}$ must be transmitted over the uplink channel. Let $\gamma_{i,t}$ denote the instantaneous signal-to-noise ratio (SNR) for user $i$ at slot $t$. A transmission is considered feasible only if $\gamma_{i,t} \ge \gamma_{\text{th}}$, where $\gamma_{\text{th}}$ is a predefined threshold for reliable communication. In practice, the
downlink is typically more reliable due to higher transmit power and beamforming at the base station; we therefore assume downlink transmission is error-free.

To model the communication delay more precisely, we characterize the uplink channel capacity. For task $i$ at slot $t$, if the channel is available, the achievable data rate is given by the Shannon capacity:
\begin{eqnarray}\label{eq:shannon}
&&\hspace{0.3cm} R_{i,t} = B_{i,t} \log_2 \left(1 + \gamma_{i,t}\right), \\
&&\hspace{0.4cm} 
\sum_{i \in \mathcal{A}_t} B_{i,t} \leq B, \notag
\end{eqnarray}
where $B_{i,t}$ is the bandwidth allocated to this transmission, subject to the total bandwidth constraint; $\mathcal{A}_t \subseteq \mathcal{U}_t$ denotes the subset of active tasks that are transmitting in slot $t$.

The size of the context to be transmitted is proportional to its length in tokens. Let $\overline{L}_{i,t}$ be the number of tokens in context $x_{i,t}$. With a tokenizer that encodes each token as a 32-bit sequence, the size in bits is $L_{i,t} = 32 \overline{L}_{i,t}$. The communication delay for offloading this step is therefore
\begin{equation}
\mathrm{TC}_{i,t} = \frac{L_{i,t}}{R_{i,t}}.
\label{eq:comm_delay}
\end{equation}
Note that the communication burden is step-dependent, as the context length (and thus the required transmission time) grows as the task progresses.
Therefore, even for the same user, the cost of communication evolves over time. When multiple users simultaneously attempt to access the server, this time-varying communication demand further increases the complexity of system coordination.

\subsubsection{Queuing and computational latency} 
The server can process at most $M$ requests in parallel. When an offloading request arrives and all $M$ processing budgets are occupied, the request must wait in a server queue. Let $\mathrm{TQ}_{i,t}$ denote the queuing delay experienced by task $i$'s request upon arrival at the server queue in slot $t$. This delay is a function of the remaining service times of the tasks currently in service and the number of requests already waiting.

Once a request is admitted for processing, the LLM requires time $\mathrm{TL}_{i,t}$ to generate the next reasoning segment. Conversely, if the step is executed locally, the SLM requires time $\mathrm{TS}_{i,t}$. These computational delays are not fixed constants; they depend on the current context length and the model architecture. A detailed characterization of these delays, accounting for the prefill and decode stages of transformer inference, is provided in Section~\ref{sec:prada}.

For a step executed locally in slot $t$, the latency is simply $\mathrm{TS}_{i,t}$. For an offloaded step, the total latency comprises all three components: $\mathrm{TL}_{i,t} + \mathrm{TC}_{i,t} + \mathrm{TQ}_{i,t}$.

We emphasize that these latency components are tightly coupled across users. A decision to offload a particular step not only affects the completion time of the current task but also alters the queuing state and bandwidth availability for all other active tasks in subsequent slots. The interplay of local computation, wireless transmission, and shared server resources under dynamic conditions renders the global scheduling problem highly complex and time-varying.

\section{The Quality-Latency Trade-off}
\label{sec:globle_formulation}

At the heart of heterogeneous agent collaboration lies a fundamental trade-off: invoking the server LLM improves reasoning quality but incurs communication, queuing, and computation delays, whereas relying on the local SLM reduces latency at the potential cost of accuracy. This tension is further amplified by the sequential nature of each task. A task unfolds over multiple reasoning steps, where the quality of later steps depends on earlier correctness, and each decision alters the context length, and thus the cost of future offloading. Hence, the central question is how to design an online policy that continuously balances quality and latency under uncertainty and resource coupling.

To evaluate any such policy in a principled way, we adopt a Lagrangian relaxation perspective. Let $U$ denote a task-level quality metric, and let $\Delta$ be the average end-to-end latency across all users. We introduce a non-negative parameter $\beta \ge 0$ that quantifies the relative penalty on latency. The system utility of a policy is then defined as the maximization objective:
\begin{equation}
\max \bigl( U - \beta \Delta \bigr).
\label{eq:system_objective}
\end{equation}
A larger $\beta$ forces the policy to favor low latency (more local steps), while a smaller $\beta$ encourages quality (more offloading). This scalar utility provides a common ground for comparing different policies across the entire Pareto frontier of the quality-latency trade-off. However, directly optimizing this objective is infeasible in practice due to the enormous state space and the coupling constraints. The remainder of this section formalizes the problem as a global Markov Decision Process (MDP), discusses why solving it exactly is intractable, and then gives an overview of our proposed approach, which efficiently approximates the optimal trade-off.

\subsection{The Global MDP}

We formalize the heterogeneous agent collaboration problem as a global MDP. This formulation serves two purposes: it provides a unified mathematical description of the system dynamics, and it explicitly exposes the coupling that makes direct optimization intractable, thereby motivating our solution.

\subsubsection{State space}
At the beginning of each system slot $t$, the system observes a global state $\mathbf{s}_t$. For each active user $i \in \mathcal{U}_t$, the local state component comprises the current context $x_{i,t}^{(k_{i,t})}$ and the instantaneous channel SNR $\gamma_{i,t}$. 
On the server side, let $\mathcal{M}_t$ denote the set of requests currently being processed by the LLM, with $|\mathcal{M}_t| \le M$. For each in-service request $j \in \mathcal{M}_t$, we maintain its remaining service time $\Gamma_{j,t}$, which can be predicted from the context length and the number of tokens yet to be generated. 
Finally, let $\mathcal{Q}_t$ represent the set of requests waiting in the server queue. The global state is thus
\begin{equation}\label{eq:global_state_new}
\mathbf{s}_t = \Bigl( \{x_{i,t}^{(k_{i,t})}, \gamma_{i,t}\}_{i\in\mathcal{U}_t},\, \{\Gamma_{j,t}\}_{j\in\mathcal{M}_t},\, \mathcal{Q}_t \Bigr).
\end{equation}

\subsubsection{Action space}
For each active user $i$, the system must choose a task-level action $\alpha_{i,t}$, where
\begin{equation}\label{eq:task_level_action_cases}
\alpha_{i,t}
=
\begin{cases}
0, & \text{processed locally by the SLM},\\[1mm]
1, & \text{offloaded to the server queue},\\[1mm]
2, & \text{admitted for immediate LLM execution}.
\end{cases}
\end{equation}

When a request is admitted for immediate execution ($\alpha_{i,t}=2$), the scheduler must also allocate a portion $B_{i,t}$ of the total bandwidth budget $B$ for transmitting the current context $x_{i,t}^{(k_{i,t})}$. The joint action in slot $t$ is therefore
\begin{equation}\label{eq:joint_action}
\boldsymbol{\alpha}_t
=
\Bigl(
\{\alpha_{i,t}\}_{i\in\mathcal{U}_t},\,
\{B_{i,t}\}_{i\in\mathcal{A}_t}
\Bigr),
\end{equation}
where $\mathcal{A}_t=\{i\in\mathcal{U}_t\mid \alpha_{i,t}=2\}$ is the set of requests admitted for immediate execution. The feasible action set must satisfy
\begin{equation}
|\mathcal{A}_t|\le M-|\mathcal{M}_t|,
\quad
\sum_{i\in\mathcal{A}_t}B_{i,t}\le B.
\label{eq:joint_action_constraints}
\end{equation}

\subsubsection{Trajectory and stepwise latency}
Consider a task $i$ that completes after $K_i$ reasoning steps (i.e., after generating $K_i$ segments). Its trajectory is a sequence of decision-state pairs:
\begin{equation}
\tau_i = \bigl( x_{i,t_0}^{(0)}, \alpha_{i,t_0}, x_{i,t_1}^{(1)}, \alpha_{i,t_1}, \dots, x_{i,t_{K_i}}^{(K_i)} \bigr),
\label{eq:trajectory_new}
\end{equation}
where $t_k$ denotes the system slot at which the decision for step $k$ is made and $x_{i,t_{K_i}}^{(K_i)}$ is the terminal context. The total latency accumulated by task $i$ along this trajectory is
\begin{equation}
\Delta(\tau_i) = \sum_{k=0}^{K_i-1} \delta_{i,t_k},
\label{eq:trajectory_latency_new}
\end{equation}
where $\delta_{i,t_k}$ is the step latency incurred by the decision made at slot $t_k$. Using the delay components defined in Section~\ref{sec:system_model}, we have
\begin{equation}
\delta_{i,t_k}
\!=\!
\begin{cases}
\mathrm{TS}_{i,t_k}, \!& \alpha_{i,t_k}\!=\!0,\\[1mm]
\mathrm{TL}_{i,t_k} + \mathrm{TC}_{i,t_k} + \mathrm{TQ}_{i,t_k}, \!& \alpha_{i,t_k}\!=\!1,\\[1mm]
\mathrm{TL}_{i,t_k} + \mathrm{TC}_{i,t_k}, \!& \alpha_{i,t_k}\!=\!2.
\end{cases}
\label{eq:step_latency_under_h}
\end{equation}

For each trajectory, we assign a return that measures the improvement in reasoning quality according to the objective in \eqref{eq:system_objective}. The global MDP can then be solved by maximizing this per-trajectory return.

\subsubsection{Why direct optimization is intractable}
\label{subsec:difficulties}

Solving the global MDP directly, as a centralized solution, faces several fundamental obstacles.
\begin{itemize}[leftmargin=0.5cm]
    \item First, the joint action space couples heterogeneous decisions such as per-slot binary scheduling (local vs. offload), server admission control (queue vs. immediate execution), and continuous bandwidth allocation. These factors are strongly coupled. A decision made for one user at one generation step may affect not only its own quality and delay, but also the queuing state and bandwidth availability experienced by other users in subsequent slots. This makes exact optimization infeasible. 
    \item Second, even if a heuristic centralized solver were employed, it would necessitate that every active user uploads its complete context $x_{i,t}^{(k_{i,t})}$ to the server at the beginning of each slot purely for the purpose of decision making. Given that context sizes $\overline{L}_{i,t}$ grow with reasoning progress, this overhead would consume a significant portion of the uplink budget $B$ before any actual offloading occurs, defeating the purpose of communication-efficient acceleration.
    \item  Third, the resulting mixed-integer program with time-varying constraints offers poor scalability and high latency. The centralized scheduling quickly becomes computationally intractable as the number of users grows. 
\end{itemize}

These practical difficulties motivate our PRM-aided two-stage decoupled acceleration (PRADA) framework.

\subsection{Key principles of PRADA}

PRADA is underpinned by three core insights: two-stage decoupling, PRM as an offline supervisor, and precise characterization of computation delay.

\subsubsection{Two-stage decoupling}

To addresses the challenges in Section~\ref{subsec:difficulties}, we decouple the global decision into two distinct stages with separate scopes of responsibility:

\textit{Stage 1} (decentralized edge screening): 
At the beginning of slot $t$, each active user $i \in \mathcal{U}_t$ independently runs a lightweight local policy that depends solely on its own current context $x_{i,t}^{(k_{i,t})}$. This policy produces a binary candidate indicator: whether the upcoming step $k_{i,t}$ is a candidate for server offloading or should be processed locally. No coordination or resource negotiation occurs at this stage. The policy is designed to be extremely small, ensuring that its inference overhead is negligible compared to the subsequent generation step.

\textit{Stage 2} (centralized server scheduling): The server collects candidate requests from those users that pass the first-stage screening. 
Let $\mathcal{O}_t \subseteq \mathcal{U}_t$ denote this candidate set. Depending on the server concurrency limit $M$ and and bandwidth budget $B$, $|\mathcal{O}_t|$ is typically much smaller than $|\mathcal{U}_t|$. 
The server then solves a resource-constrained optimization problem to assign final actions to each candidate request, deciding whether to admit it for immediate LLM execution, place it in a waiting queue, or return it for local SLM execution. This results in a low-complexity threshold policy derived from Lagrangian relaxation (detailed later).

This two-stage decoupling offers immediate practical advantages. First, the original coupled, mixed-integer global optimization is replaced by a set of lightweight per-user evaluations and a manageable server-side knapsack-like problem. Second, and more importantly, the system \textit{does not} require all active users to upload their contexts for decision making; only the contexts of the screened candidates are transmitted, and only when they are actually scheduled for offloading. The resulting reduction in uplink traffic is essential for maintaining low communication delay $\mathrm{TC}_{i,t}$.

\subsubsection{PRM as an offline supervisor}

A central challenge in any sequential decision problem is how to evaluate the quality of an action taken at an intermediate step. In our setting, the system must judge whether a particular offloading decision is likely to lead to a correct final answer. To this end, we employ PRMs. A PRM is a neural network that takes an intermediate reasoning state, namely the current context $x_{i,t}^{(k_{i,t})}$, and outputs a scalar score estimating the probability that this state will eventually lead to a correct final answer. PRMs have been used in prior collaborative reasoning works (e.g., RSD \cite{liao2025rewardguided}, G-Boost \cite{fan2025g}) as an online component: at each step, the PRM is invoked to compare the quality of continuing with the SLM versus switching to the LLM.

However, we observe that this online usage of PRM introduces a fundamental inefficiency. The PRM itself is typically as large as the LLM (often several billion parameters), and invoking it at every step for every user adds substantial latency and computational overhead, exactly the opposite of what acceleration aims to achieve. In a multi-user environment where server resources are already constrained by the LLM inference workload, the additional burden of frequent PRM queries would only exacerbate the queuing delay $\mathrm{TQ}_{i,t}$ and compete for the same scarce computational budget.

Our insight is that the per-user decision is binary (stay local or request server). For such a simple binary choice, a very lightweight decision network, perhaps with only a few hundred thousand parameters, is sufficient, provided it receives proper supervision during an offline training phase. Moreover, the computationally expensive PRM can be used as a {teacher} during this offline phase, providing dense reward signals without any impact on online inference latency.

Consequently, PRADA changes the role of the PRM from an online evaluator to an offline supervisor. This design completely eliminates the PRM's inference latency and memory consumption from the online execution path, while still allowing the system to benefit from the PRM's fine-grained semantic quality assessment during training. The result is a lightweight, fast, yet accurate edge decision mechanism that is ideally suited for the decentralized first stage.

\subsubsection{Precise characterization of computation delay}

A prerequisite for optimizing the trade-off in \eqref{eq:system_objective} is the ability to express all relevant costs in a common unit. While communication delay $\mathrm{TC}_{i,t}$ and queuing delay $\mathrm{TQ}_{i,t}$ are naturally measured in seconds, the computational work performed by SLMs and LLMs is most accurately quantified in floating-point operations (FLOPs). Existing approaches to collaborative inference often handle this discrepancy in an ad-hoc or imprecise manner. For instance, Hybrid SLM-LLM (HSL)~\cite{DBLP:conf/edgefm/HaoJJ0C24} approximates computational delay using only the model size, assigning a fixed delay cost of $10$ to a $10$B model and $1$ to a $1$B model without accounting for the variability induced by context length or the number of generated tokens. Other frameworks, such as RSD~\cite{liao2025rewardguided} and Hierarchical Language Model (HLM)~\cite{11140540}, do not explicitly incorporate latency into their inference-time decisions; instead, they aim to reduce latency indirectly by sacrificing a portion of the achievable accuracy.

In contrast, PRADA adopts a rigorous, context-aware model of computational delay. This precise characterization is essential because the durations $\mathrm{TL}_{i,t}$ and $\mathrm{TS}_{i,t}$ depend strongly on both the current context length $\overline{L}_{i,t}$ and the number of new tokens $m_{\text{new}}$ to be generated in the current reasoning step. Moreover, an accurate delay model allows the server scheduler to predict the remaining service times of in-flight requests, which is critical for estimating queuing delays $\mathrm{TQ}_{i,t}$.

Modern SLMs and LLMs are predominantly based on the Transformer architecture. As illustrated in Figure~\ref{fig:pre_dec}, their inference process consists of two distinct phases:
\begin{itemize}[leftmargin=0.5cm]
    \item \textit{Prefill phase:} The model processes the entire input context $x_{i,t}^{(k_{i,t})}$ in parallel. During this phase, the key-value (KV) pairs for all input tokens are computed and cached, and the first token of the response is generated.
    \item \textit{Decode phase:} Subsequent tokens are generated one by one in an autoregressive manner. Leveraging the stored KV cache, this phase computes attention only for the newest token, dramatically reducing the per-token computational cost.
\end{itemize}

\begin{figure}[t]
    \centering
    \includegraphics[width=0.5\textwidth]{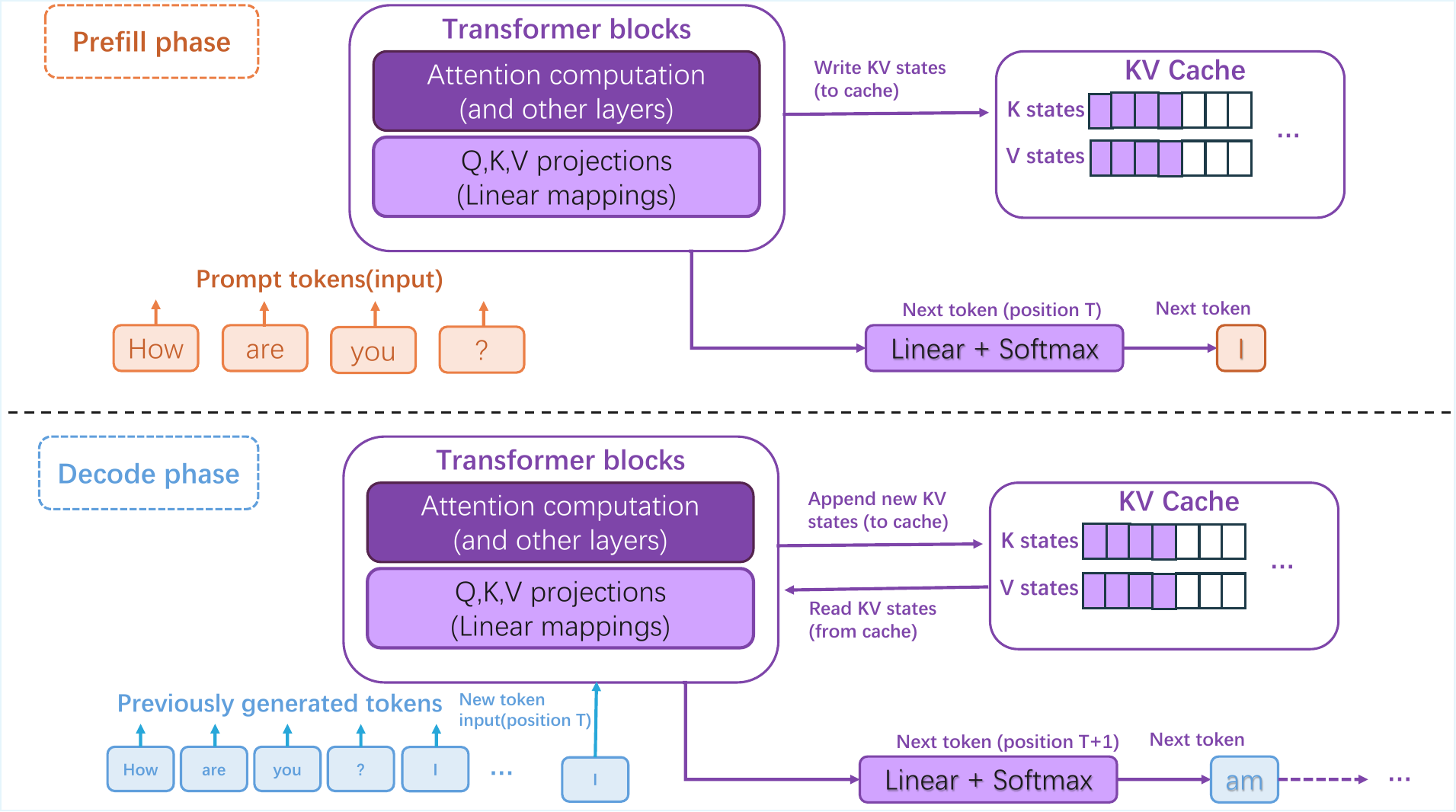}
    \caption{Illustration of the prefill and decode phases in model inference.}
    \label{fig:pre_dec}
\end{figure}

To quantify the required computation, let $m = \overline{L}_{i,t}$ denote the number of tokens in the context $x_{i,t}^{(k_{i,t})}$, let $o$ be the hidden dimension of the model, and let $v$ be the number of decoder layers. The total FLOPs for the prefill phase is well-approximated by
\begin{equation}\label{eq:prefill_cost}
f_{\text{pre}}(m,o,v) = v\left(2m^2o + 2mo + 4mo^2\right).
\end{equation}
In this expression, the term $2m^2o$ accounts for the attention score computation (query-key dot products and value weighting), while $2mo$ and $4mo^2$ correspond to the Softmax normalization and the linear projections, respectively.

After the prefill phase, each subsequent decode step for a single new token incurs
\begin{equation}
f_{\text{dec}}(m,o,v) \approx v\left(o^2 + mo\right),
\label{eq:decode_cost}
\end{equation}
where $o^2$ captures the query projection and feed-forward computation for the new token, and $mo$ accounts for the attention operation against the cached keys and values of length $m$.

In our stepwise framework, a single reasoning step may generate a segment comprising $m_{\text{new}}$ new tokens (e.g., a full sentence or a reasoning fragment). The incremental FLOPs cost for this step, building upon the prefill computation already performed (which is incurred only once per task), is given by
\begin{equation}
f_{\text{step}}(m,o,v,m_{\text{new}}) = m_{\text{new}}o^2 + \frac{o\,m_{\text{new}}(2m + m_{\text{new}} - 1)}{2}.
\label{eq:step_cost}
\end{equation}
This expression captures the fact that the cost of generating $m_{\text{new}}$ tokens depends on both the accumulated context length $m$ and the length of the new segment.

Finally, to convert FLOPs into a time delay, we use the sustained computational capability of the respective hardware. Let $\text{FlopsLLM}$ and $\text{FlopsSLM}$ denote the peak FLOPs per second achievable by the server and the edge devices, respectively. The computational latency for a step executed by the server LLM is
\begin{equation}
\mathrm{TL}_{i,t} = \frac{ f_{\text{step}}(m, o_{\text{L}}, v_{\text{L}}, m_{\text{new}})}{\text{FlopsLLM}},
\label{eq:TL_delay}
\end{equation}
and similarly, for the local SLM,
\begin{equation}
\mathrm{TS}_{i,t} = \frac{f_{\text{step}}(m, o_{\text{S}}, v_{\text{S}}, m_{\text{new}})}{\text{FlopsSLM}}.
\label{eq:TS_delay}
\end{equation}

With this precise characterization, the three principal delay components: communication $\mathrm{TC}_{i,t}$, queuing $\mathrm{TQ}_{i,t}$, and computation $\mathrm{TL}_{i,t}$ or $\mathrm{TS}_{i,t}$, are all expressed in seconds. This unified time metric enables a consistent and principled optimization of the composite objective in \eqref{eq:system_objective}, where the trade-off parameter $\beta$ directly balances the quality gain of using the LLM against the additional latency it incurs.

\section{Two-Stage Decoupled Collaboration}
\label{sec:prada}

As stated in Section~\ref{sec:globle_formulation}, PRADA decomposes the global sequential decision problem into two sequential stages with distinct responsibilities and information requirements, enabling scalable and communication-efficient collaboration between edge SLMs and the server LLM.

\subsection{Decentralized Edge Screening (Stage 1)}

At the beginning of slot $t$, every active user $i \in \mathcal{U}_t$ is about to initiate its next reasoning step $k_{i,t}$ with current context $x_{i,t}^{(k_{i,t})}$. The role of the first stage is to answer a simple question: given the present context, is it worth even considering the server LLM for this step, or should the step be executed locally without further deliberation?

\subsubsection{Local decision}
The first stage outputs a binary auxiliary action
\begin{equation}
a_{i,t} \in \{0,1\},
\label{eq:pre_action_new}
\end{equation}
where $a_{i,t}=0$ indicates the user will remain local, and $a_{i,t}=1$ submits the request as a candidate for server scheduling. 

Crucially, this decision is made using only information that is locally available to the user. The policy is parameterized by a lightweight neural network $\pi_\theta(a_{i,t} \mid x_{i,t}^{(k_{i,t})})$ that maps the current context to a probability distribution over the two actions. Since no information exchange among users or with the server is required, the screening can be performed asynchronously and with negligible computational overhead. This design principle directly addresses the communication and scalability concerns raised in Section~\ref{sec:globle_formulation}.

It is important to emphasize that $a_{i,t}=1$ does not constitute a commitment to server execution. It merely marks the request as a candidate; the final action $\alpha_{i,t} \in {0,1,2}$ will be determined by the second-stage scheduler under explicit resource constraints. In contrast, $a_{i,t}=0$ is binding and immediately sets $\alpha_{i,t}=0$, removing the request from the candidate pool. Consequently, the set of requests that advance to Stage 2 is $\mathcal{O}_t = \{i \in \mathcal{U}_t \mid a_{i,t}=1\}$, and its size $|\mathcal{O}_t|$ is typically a small fraction of $|\mathcal{U}_t|$, which is the key to making the subsequent resource allocation tractable.

\subsubsection{Offline training with PRM supervision}
The screening policy $\pi_\theta$ must be trained to make decisions that align with the global objective \eqref{eq:system_objective}. However, training directly on the full global return is complicated by the fact that the consequences of $a_{i,t}$ depend on second-stage scheduling outcomes and the resulting queuing and communication delays. To circumvent this circular dependency, we train $\pi_\theta$ using a coarse approximation of the global return that isolates the computational trade-off inherent to the binary decision while deferring the resource contention aspects to Stage~2.

For a complete trajectory $\tau_i$ of task $i$, we define the coarse trajectory-level return as
\begin{equation}
G_{\text{pre}}(\tau_i) = \mathrm{PRM}\bigl(x_{i,t_{K_i}}^{(K_i)}\bigr) - \mathrm{PRM}\bigl(x_{i,t_0}^{(0)}\bigr) - \beta\sum_{t \in \mathcal{T}_i} D\bigl(x_{i,t}^{(k_{i,t})}, a_{i,t}\bigr),
\label{eq:pre_return_new}
\end{equation}
where $\mathcal{T}_i$ is the set of decision slots at which task $i$ was active, and $D(x, a)$ is the computational cost incurred by the action: $D(x_{i,t},0) = \mathrm{TS}_{i,t}$ and $D(x_{i,t},1) = \mathrm{TL}_{i,t}$. By training $\pi_\theta$ to maximize $G_{\text{pre}}$, we encourage it to learn which steps intrinsically benefit from the stronger reasoning capabilities of the LLM.

To enable efficient policy gradient estimation, we decompose the trajectory-level return into a sum of per-step rewards. Specifically, we define the immediate reward associated with the transition from context $x_{i,t}^{(k_{i,t})}$ to the context at the next decision epoch (which occurs when the step completes and the task becomes active again) as
\begin{equation}
r_{i,t} = \mathrm{PRM}\bigl(x_{i,t'}^{(k_{i,t}+1)}\bigr) - \mathrm{PRM}\bigl(x_{i,t}^{(k_{i,t})}\bigr) - \beta D\bigl(x_{i,t}^{(k_{i,t})}, a_{i,t}\bigr),
\label{eq:pre_reward_new}
\end{equation}
where $t' > t$ is the slot at which the task next becomes active. With this definition, the coarse return satisfies
\begin{equation}
G_{\text{pre}}(\tau_i) = \sum_{t \in \mathcal{T}_i} r_{i,t},
\label{eq:pre_return_sum_new}
\end{equation}
which transforms the problem into a standard sequential decision format.

The policy $\pi_\theta$ and an associated value network $V_\varphi$ (used for variance reduction) are trained offline using Proximal Policy Optimization (PPO) \cite{2017arXiv170706347S}. During training, a three-step temporal-difference target is employed to estimate the advantage:
\begin{equation}
\widetilde{G}_{i,t} = r_{i,t} + r_{i,t'} + r_{i,t''} + V_\varphi\bigl(x_{i,t'''}^{(k_{i,t'''})}\bigr),
\label{eq:pre_td_target_new}
\end{equation}
\begin{equation}
A_{i,t} = \widetilde{G}_{i,t} - V_\varphi\bigl(x_{i,t}^{(k_{i,t})}\bigr),
\label{eq:pre_advantage_new}
\end{equation}
where $t', t'', t'''$ denote the successive decision slots of the same task. The policy parameters $\theta$ are then updated to maximize the clipped surrogate objective of PPO, while $\varphi$ is updated to minimize the mean squared error between $V_\varphi$ and the observed returns.

Critically, the PRM is used solely as an offline supervisor during this training phase. Its computationally expensive forward passes are never performed online. Once trained, the lightweight network $\pi_\theta$ can be deployed on edge devices with minimal latency and memory footprint, perfectly fulfilling the design goal of a decentralized screening stage.

At the conclusion of Stage~1, the system has produced a candidate set $\mathcal{O}_t \subseteq \mathcal{U}_t$. All other active users ($i \in \mathcal{U}_t \setminus \mathcal{O}_t$) have been assigned $\alpha_{i,t}=0$ and will execute their current step locally without further delay. The candidate set is then passed to the centralized second stage, which is responsible for resolving the final actions.

\subsubsection{Optimality of the local decision}
We now justify why a Stage~1 decision to stay local ($a_{i,t}=0$) remains globally optimal even though the policy is trained with a coarse reward that omits communication and queuing delays. The key observation is that these omitted delays only penalize the offloading action, leaving the local action unchanged.

To connect the first-stage coarse decision with the true system objective, we now examine when a local decision made under the coarse reward remains valid after communication and queuing effects are taken into account. Let $Q_{\pi_\theta}(x_{i,t_{K_i}}^{(K_i)}, a_{i,t})$ denote the action-value function under the coarse reward defined in \eqref{eq:pre_reward_new}, and let $\widetilde{Q}_{i,t}(a_{i,t})$ denote the action-value function under the true global objective, which additionally includes the communication delay $\mathrm{TC}_{i,t}$ and the queuing delay $\mathrm{TQ}_{i,t}$.

\begin{prop}\label{thm:1}
If for a given context $x_{i,t_{K_i}}^{(K_i)}$ the coarse value function satisfies $Q_{\pi_\theta}(x_{i,t_{K_i}}^{(K_i)}, 1) \le Q_{\pi_\theta}(x_{i,t_{K_i}}^{(K_i)}, 0)$, then the true value function also satisfies $\widetilde{Q}_{i,t}(1) \le \widetilde{Q}_{i,t}(0)$. Consequently, a decision to stay local ($a_{i,t}=0$) that is optimal under the coarse reward remains optimal under the true global objective.
\end{prop}

\begin{proof}
See Appendix \ref{sec:AppA}.
\end{proof}

This proposition guarantees that training $\pi_\theta$ with the coarse reward (which ignores communication and queuing) does not lead to suboptimal local decisions. The screening policy can safely be trained offline using only computational costs, and any request that is filtered out at Stage~1 would also be rejected by an optimal policy that had full knowledge of future communication and queuing delays.

\subsection{Centralized Server Scheduling (Stage 2)}
\label{sec:secondary_decision}

Stage 1 produces a candidate set $\mathcal{O}_t$ of requests that have passed local screening. These candidates have indicated that, based solely on their local context, server offloading may be beneficial. Stage 2 now resolves the actual resource allocation: given the candidate set, the server must assign each request a final action $\alpha_{i,t} \in {0,1,2}$ (as defined in \eqref{eq:task_level_action_cases}) while respecting the server concurrency limit $M$ and the bandwidth budget $B$.

\subsubsection{Refining the value estimates}

The overall system state space and action space are identical to those defined in Section~\ref{sec:globle_formulation}. The difference in Stage 2 is that the scheduler only considers requests within the candidate set $\mathcal{O}_t$. At slot $t$, the state and action are refined as follows:
\begin{equation*}\label{eq:global_state_2}
\overline{\mathbf{s}}_t = \Bigl( \{x_{i,t}^{(k_{i,t})}, \gamma_{i,t}\}_{i\in\mathcal{O}_t},\, \{\Gamma_{j,t}\}_{j\in\mathcal{M}_t},\, \mathcal{Q}_t \Bigr).
\end{equation*}
\begin{equation*}\label{eq:joint_action_2}
\overline{\boldsymbol{\alpha}}_t
=
\Bigl(
\{\alpha_{i,t}\}_{i\in\mathcal{O}_t},\,
\{B_{i,t}\}_{i\in\mathcal{A}_t}
\Bigr).
\end{equation*}

Here, $\mathcal{Q}_t$ specifically includes the time each request has spent waiting as well as the expected server-side computation time if processed by the LLM. For each candidate we also maintain two auxiliary variables:
\begin{itemize}[leftmargin=0.5cm]
    \item $\tau_{i,t}$: the time the request has already spent waiting in the server queue before slot $t$ (zero if the request is newly submitted);
    \item $d_{i,t}$: the predicted server-side computation time $\mathrm{TL}_{i,t}$, which is a function of the context length and the number of new tokens to be generated.
\end{itemize}

Thus, the system state can also be written as
\[
\overline{\mathbf{s}}_t = \Bigl( \{x_{i,t}^{(k_{i,t})}, \gamma_{i,t}\}_{i\in\mathcal{O}_t},\, \{\Gamma_{j,t}\}_{j\in\mathcal{M}_t},\, \{\tau_{i,t}, d_{i,t}\}_{i\in\mathcal{O}_t} \Bigr).
\]
Note that the waiting time $\tau_{i,t}$ reflects past queuing delay, whereas any future queuing delay incurred if the request is placed in the queue will be denoted by $\mathrm{TQ}_{i,t}$ and is distinct from $\tau_{i,t}$.

If a candidate is admitted for immediate execution ($\alpha_{i,t}=2$), the scheduler allocates a dedicated bandwidth slice $B_{i,t}$ for its context transmission. According to \eqref{eq:comm_delay}, the resulting communication delay is
\begin{equation}
\mathrm{TC}_{i,t} = \frac{32\,\overline{L}_{i,t}}{B_{i,t}\log_2(1+\gamma_{i,t})}.
\label{eq:tc_bandwidth_function}
\end{equation}
For a request that is placed in the queue ($\alpha_{i,t}=1$), it will eventually be served but will incur a queuing delay $\mathrm{TQ}_{i,t}$ whose magnitude depends on the current server state $\mathbf{c}_t=\{\Gamma_{i,t}\}_{j \in \mathcal{O}_t}$ and the set of competing candidates.

To make the scheduling decision, we need a per-request metric that captures the net benefit of choosing each action. Recall from the first-stage training that the policy $\pi_\theta$ is associated with state-action value functions $Q_{\pi_\theta}(x_{i,t}^{(k_{i,t})},a_{i,t})$ for the binary actions $a_{i,t} \in \{0,1\}$. These functions estimate the expected cumulative coarse return conditioned on taking action $a_{i,t}$ in context $x_{i,t}^{(k_{i,t})}$ and following $\pi_\theta$ thereafter. The second stage refines these estimates by accounting for the resource costs that were omitted in the coarse model. Specifically, for each candidate $i$, we define a modified immediate reward that incorporates communication and queuing penalties:
\begin{eqnarray}\label{eq:modified_reward_stage2}
&&\hspace{-0.5cm} \widetilde{r}_{i,t}(\alpha_{i,t}) = \\
&&\hspace{-0.5cm} \begin{cases}
r_{i,t}\bigl(x_{i,t}^{(k_{i,t})}, 0\bigr) - \beta\,\tau_{i,t}, & \hspace{-0.2cm} \alpha_{i,t} = 0, \\[1mm]
r_{i,t}\bigl(x_{i,t}^{(k_{i,t})}, 1\bigr) \!-\! \beta\,\bigl(\mathrm{TC}_{i,t} \!+\! \mathrm{TQ}_{i,t}\bigr) \!-\! \beta\,\tau_{i,t}, & \hspace{-0.2cm} \alpha_{i,t} = 1, \\[1mm]
r_{i,t}\bigl(x_{i,t}^{(k_{i,t})}, 1\bigr) - \beta\,\mathrm{TC}_{i,t} - \beta\,\tau_{i,t}, & \hspace{-0.2cm} \alpha_{i,t} = 2.\notag
\end{cases} 
\end{eqnarray}
The trade-off parameter $\beta$ (introduced in \eqref{eq:system_objective}) explicitly scales the latency penalties to match the units of the PRM score. The constant offset $\beta\tau_{i,t}$ subtracts the accumulated waiting time that the request has already endured, thereby isolating the incremental future cost of the current decision.

Correspondingly, we define a modified action-value function $\widetilde{Q}_{i,t}(\alpha_{i,t})$ that substitutes the appropriate cost into the first-stage $Q$-values:
\begin{eqnarray}\label{eq:modified_q_stage2}
&&\hspace{-0.6cm} \widetilde{Q}_{i,t}(\alpha_{i,t}) =  \\
&&\hspace{-0.6cm} \begin{cases}
Q_{\pi_\theta}\!\bigl(x_{i,t}^{(k_{i,t})}, 0\bigr) - \beta\,\tau_{i,t}, & \hspace{-0.2cm} \alpha_{i,t} = 0, \\[1mm]
Q_{\pi_\theta}\!\bigl(x_{i,t}^{(k_{i,t})}, 1\bigr) \!-\! \beta\,\bigl(\mathrm{TC}_{i,t} \!+\! \mathrm{TQ}_{i,t}\bigr) \!-\! \beta\,\tau_{i,t}, & \hspace{-0.2cm} \alpha_{i,t} = 1, \\[1mm]
Q_{\pi_\theta}\!\bigl(x_{i,t}^{(k_{i,t})}, 1\bigr) - \beta\,\mathrm{TC}_{i,t} - \beta\,\tau_{i,t}, & \hspace{-0.2cm} \alpha_{i,t} = 2.\notag
\end{cases}
\end{eqnarray}
Because the first-stage value functions already capture the long-term semantic and computational trade-offs, the second-stage scheduler can focus on resolving the instantaneous resource competition using the per-slot aggregated utility $\sum_{i\in\mathcal{O}_t} \widetilde{Q}_{i,t}(\alpha_{i,t})$.

\subsubsection{Optimization}

At each slot $t$, the server faces the following constrained optimization problem:
\begin{equation}
\begin{aligned}
\max_{\overline{\boldsymbol{\alpha}}_t} \quad & \sum_{i\in\mathcal{O}_t} \widetilde{Q}_{i,t}(\alpha_{i,t}) \\
\text{s.t.} \quad & \sum_{i\in\mathcal{O}_t} \mathbb{I}(\alpha_{i,t}=2) \le M - |\mathcal{M}_t|, \\
& \sum_{i\in\mathcal{O}_t} B_{i,t}\,\mathbb{I}(\alpha_{i,t}=2) \le B,
\end{aligned}
\label{eq:stage2_optimization}
\end{equation}
where $\mathbb{I}(\cdot)$ is standard indicator function.

The first constraint ensures that the number of immediately admitted requests does not exceed the available server concurrency, while the second enforces the total bandwidth budget. This problem couples the discrete admission choices and the continuous bandwidth variables across all candidates.

To decouple the problem, we introduce two non-negative Lagrange multipliers: $\lambda_s \ge 0$ for the server concurrency constraint, and $\mu \ge 0$ for the bandwidth constraint. The Lagrangian function is
\begin{equation}
\begin{aligned}
\mathcal{L} = \sum_{i\in\mathcal{O}_t} \widetilde{Q}_{i,t}(\alpha_{i,t}) 
&- \lambda_s \Bigl( \sum_{i\in\mathcal{O}_t} \mathbb{I}(\alpha_{i,t}=2) - (M - |\mathcal{M}_t|) \Bigr) \\
&- \mu \Bigl( \sum_{i\in\mathcal{O}_t} B_{i,t}\,\mathbb{I}(\alpha_{i,t}=2) - B \Bigr).
\end{aligned}
\label{eq:lagrangian_stage2}
\end{equation}
Rearranging the terms isolates the per-request contributions:
\begin{eqnarray*}\label{eq:lagrangian_separated}
&&\hspace{-0.6cm} \mathcal{L} = \sum_{i\in\mathcal{O}_t} \Bigl[ \widetilde{Q}_{i,t}(\alpha_{i,t}) - \lambda_s \mathbb{I}(\alpha_{i,t}=2) - \mu B_{i,t} \mathbb{I}(\alpha_{i,t}=2) \Bigr] + \\
&&\hspace{0.2cm} \lambda_s (M - |\mathcal{M}_t|) + \mu B.
\end{eqnarray*}
Since the last two terms are independent of the decision variables, the optimal actions for a fixed pair $(\lambda_s, \mu)$ can be determined by comparing the modified values on a per-request basis.

\subsubsection{Optimal bandwidth allocation for immediate admission}

For a candidate that is selected for immediate execution ($\alpha_{i,t}=2$), the optimal bandwidth $B_{i,t}^*$ is obtained by maximizing the per-request Lagrangian term. Substituting $\widetilde{Q}_{i,t}(2)$ from \eqref{eq:modified_q_stage2} and differentiating with respect to $B_{i,t}$ yields
\begin{eqnarray*}
&&\hspace{-0.5cm} \frac{\partial}{\partial B_{i,t}} \Bigl[ Q_{\pi_\theta}\!\bigl(x_{i,t}^{(k_{i,t})}, 1\bigr) \!-\! \beta\,\mathrm{TC}_{i,t}(B_{i,t}) \!-\! \beta\tau_{i,t} \!-\! \lambda_s \!-\! \mu B_{i,t} \Bigr] \\
&&\hspace{-0.5cm} = \beta \frac{32\,\overline{L}_{i,t}}{B_{i,t}^2 \log_2(1+\gamma_{i,t})} - \mu = 0.
\end{eqnarray*}
Solving for $B_{i,t}$ gives the conditional optimal bandwidth
\begin{equation}
B_{i,t}^* = \sqrt{\frac{32\beta\,\overline{L}_{i,t}}{\mu \log_2(1+\gamma_{i,t})}}.
\label{eq:optimal_bandwidth_stage2}
\end{equation}
Equation \eqref{eq:optimal_bandwidth_stage2} reveals an intuitive allocation rule: requests with larger context sizes (higher $\overline{L}_{i,t}$) or poorer channel conditions (lower $\gamma_{i,t}$) receive a larger share of the bandwidth budget. The multiplier $\mu$ acts as a uniform price that balances the marginal benefit of reduced transmission time against the cost of consuming shared bandwidth.

\subsubsection{Threshold-based scheduling policy}

Because the action $\alpha_{i,t}$ is discrete, the optimal choice is found by direct comparison of the Lagrangian terms rather than differentiation. We now formalize the structure of the optimal action in Stage~2.  

Define two key quantities that measure the effective gain of offloading under resource prices and is convenient for derivation:
\begin{align}
w_{i,t} &=\! \Bigl[ Q_{\pi_\theta}\!\bigl(x_{i,t}^{(k_{i,t})}, 1\bigr) \!\!-\!\! Q_{\pi_\theta}\!\bigl(x_{i,t}^{(k_{i,t})}, 0\bigr) \Bigr] \!\!-\! \!\beta\,\mathrm{TC}_{i,t} \!\!-\!\! \mu B_{i,t}, \label{eq:w_new}\\
\overline{w}_{i,t} &= \beta\,\mathrm{TQ}_{i,t} - \mu B_{i,t}. \label{eq:wbar_new} 
\end{align}
Here, $w_{i,t}$ represents the net advantage of immediate server execution over local execution after subtracting communication cost and the bandwidth usage price, while $\overline{w}_{i,t}$ compares the cost of queuing (via $\mathrm{TQ}_{i,t}$) to the price of immediate admission.

\begin{thm}[Optimal scheduling policy]
\label{thm:stage2_policy}
For each candidate request $i \in \mathcal{O}_t$ at system slot $t$, the optimal action $\alpha_{i,t}^* \in \{0,1,2\}$ that maximizes the Lagrangian (and thus the original constrained objective) is given by the following threshold rules:

\begin{enumerate}
\item {Competitive admission} ($|\mathcal{O}_t| > M - |\mathcal{M}_t|$):
\[
\alpha_{i,t}^* =
\begin{cases}
2\;(\text{immediate}), & \min(\overline{w}_{i,t},\, w_{i,t}) \ge \lambda_s,\\[2mm]
1\;(\text{queue}),      & \overline{w}_{i,t} < \lambda_s \;\text{and}\; w_{i,t} \ge \overline{w}_{i,t},\\[2mm]
0\;(\text{local}),      & \text{otherwise}.
\end{cases}
\]

\item {Server fully occupied} ($|\mathcal{M}_t| = M$):
\[
\alpha_{i,t}^* =
\begin{cases}
1\;(\text{queue}), & w_{i,t} \ge \overline{w}_{i,t},\\
0\;(\text{local}), & \text{otherwise}.
\end{cases}
\]

\item {Abundant server capacity} ($|\mathcal{O}_t| \le M - |\mathcal{M}_t|$):
\[
\alpha_{i,t}^* =
\begin{cases}
2\;(\text{immediate}), & w_{i,t} \ge 0,\\
0\;(\text{local}),     & \text{otherwise}.
\end{cases}
\]
\end{enumerate}
\end{thm}

\begin{proof}
See Appendix \ref{sec:AppB}.
\end{proof}

The theorem shows that the Stage~2 scheduler admits a threshold-based structure with respect to the effective offloading gain. In particular, sufficiently valuable requests are immediately admitted to the server, requests with moderate gain are placed in the queue when immediate admission is too costly, and the remaining requests are returned to local execution.

\subsubsection{Practical implementation}

To execute the above rules online, the server must efficiently estimate three quantities for each candidate: the value difference $Q_{\pi_\theta}(x_{i,t}^{(k_{i,t})},1) - Q_{\pi_\theta}(x_{i,t}^{(k_{i,t})},0)$, the queuing delay $\mathrm{TQ}_{i,t}$, and the server computation time $d_{i,t}$.

\textit{Value difference estimation.}
Instead of maintaining a full copy of the value network $V_\varphi$, the server employs a lightweight advantage predictor $A_\omega$ that directly estimates the advantage of the server action $A_{\pi_\theta}(x_{i,t}^{(k_{i,t})},1)$. Using the relationship between advantage functions for a binary policy, one can show that
\begin{equation}
Q_{\pi_\theta}(x_{i,t}^{(k_{i,t})},1) - Q_{\pi_\theta}(x_{i,t}^{(k_{i,t})},0) = \frac{A_{\pi_\theta}(x_{i,t}^{(k_{i,t})},1)}{\pi_\theta(0 \mid x)}.
\label{eq:adv_identity_stage2}
\end{equation}
Hence, with a compact network $A_\omega$ trained offline to predict $A_{\pi_\theta}(x_{i,t}^{(k_{i,t})},1)$, the required value difference is obtained via a single forward pass, avoiding any costly PRM evaluation online.

\textit{Queuing delay prediction.}
The future queuing delay $\mathrm{TQ}_{i,t}$ for a request placed in the queue depends on the remaining service times of the $|\mathcal{M}_t|$ in-flight LLM requests($\{\Gamma_{j,t}\}_{j\in\mathcal{M}_t}$) and on the set of other candidate requests that may be queued ahead of it. In particular, the estimated waiting time is determined by how these existing and newly queued requests occupy the limited server slots over time. Since the priority order of candidates is determined by their own $\min(\overline{w}_{i,t}, w_{i,t})$ values, which in turn depend on $\mathrm{TQ}_{i,t}$, a self-consistency loop arises. In practice, the server resolves this by initializing a candidate order (e.g., by descending $w_{i,t}$) and iteratively updating the estimated service start times until the ordering stabilizes. The server-side LLM computation time of each candidate request $d_{i,t}$ is then incorporated after the service start time is determined. Convergence is typically achieved within a few iterations due to the small size of $\mathcal{O}_t$.

\begin{rem}[Special case: fully occupied server.]
When the server is fully occupied, immediate admission is infeasible, and the scheduler only needs to compare queue admission with local fallback. In this case, we introduce the zero-price offloading gain
\begin{equation*}
w_{i,t}^0
=
Q_{\pi_\theta}\bigl(x_{i,t}^{(k_{i,t})},1\bigr)
-
Q_{\pi_\theta}\bigl(x_{i,t}^{(k_{i,t})},0\bigr)
-
\beta\,\mathrm{TC}_{i,t}\big|_{B_{i,t}=B}.
\end{equation*}
Here, $B_{i,t}=B$ is used as a fixed reference bandwidth rather than an optimized allocation result, since bandwidth is no longer an active decision variable in the fully occupied case. Therefore, $w_{i,t}^0$ should be understood as a nominal offloading gain for comparing queue admission with local execution under zero bandwidth price.
\end{rem}

\textit{Computation delay prediction.}
The server-side computation delay $d_{i,t} = \mathrm{TL}_{i,t}$ is predicted by a lightweight neural network $D_\psi$. This network takes as input the context length $\overline{L}_{i,t}$ and outputs an estimate of the required time. The network $D_\psi$ is trained offline on execution traces collected from the server LLM using a simple mean squared error loss.

With these components, the second-stage scheduler operates with minimal computational overhead. The iterative search for $\mu$ converges rapidly, and the per-request evaluations involve only small neural networks. This design ensures that the centralized scheduling stage itself does not become a bottleneck, thereby preserving the end-to-end acceleration promised by the PRADA framework.

\section{Simulation Experiments and Discussions}
\label{sec:simulation_results}

This section conducts simulation experiments to verify the effectiveness of the PRADA framework and to analyze its behavior under different resource conditions. We first evaluate the trained policy network $\pi_\theta$ in a single-user setting, isolating the quality-latency trade-off it learns before any multi-user resource contention is introduced. Building on this, we then assess the complete PRADA algorithm in dynamic multi-user scenarios, examining its ability to generalize across diverse reasoning benchmarks and its sensitivity to the two key system resources: the server parallel capacity $M$ and the total communication bandwidth $B$. Throughout all experiments, the SLM is realized by \texttt{Qwen2.5-Math-1.5B-Instruct} and the LLM by \texttt{Qwen2.5-Math-7B-Instruct}~\cite{yang2024qwen25mathtechnicalreportmathematical}.

\subsection{Single-User Evaluation of the Policy Network $\pi_\theta$}

The policy network $\pi_\theta$ forms the core of the first-stage decentralized screening and directly determines the quality-latency trade-off that will be available to the multi-user scheduler. It is therefore essential to scrutinize its behavior in isolation, before integrating it into the full PRADA system.

\begin{figure*}[t]
\centering
\includegraphics[width=0.8\textwidth]{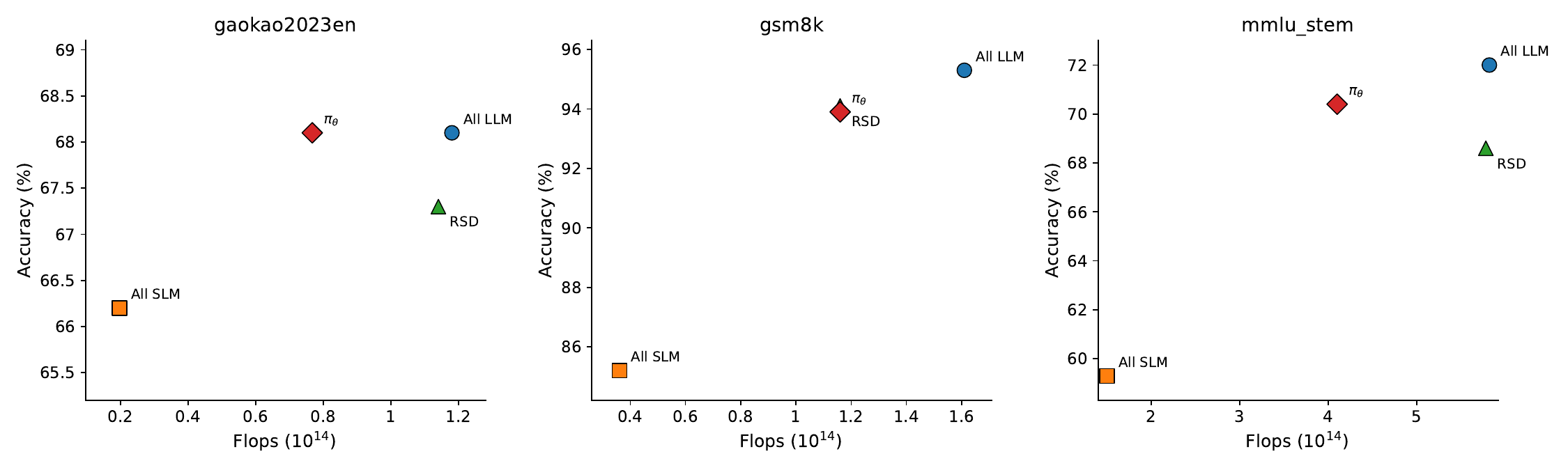}
\caption{Accuracy and computational cost of the policy network $\pi_\theta$ on three reasoning benchmarks, compared with all-SLM, all-LLM, and RSD baselines.}
\label{fig:pi_theta_all}
\end{figure*}

\subsubsection{Experimental setup and training details}
Both $\pi_\theta$ and its companion value network $V_\varphi$ are implemented as compact fully-connected networks with four layers interleaved with ReLU activations. The hidden dimensions of the first three layers are $512$, $256$, and $512$. The output layer of $\pi_\theta$ produces a two-dimensional softmax over the action set ${\text{SLM},\text{LLM}}$, while that of $V_\varphi$ returns a scalar state value. 
The networks are trained with the Adam optimizer on the \texttt{math500-test} dataset, using \texttt{Skywork-o1-Open-PRM-Qwen-2.5-7B}~\cite{he_2024_16998085} to supply the dense reward signal defined in \eqref{eq:pre_reward_new}. 
We deliberately evaluate on the test set rather than on the training set to avoid overly optimistic scores, because the SLM and the LLM have already been fine-tuned on the training portions of many standard benchmarks. The complete set of training hyperparameters is listed in Table~\ref{tab:hyper1}.

\begin{table}[t]
\caption{Hyperparameter settings for training $\pi_\theta$ and $V_\varphi$.}
\centering
\setlength{\tabcolsep}{3mm}
\begin{tabular}{ccc}
\toprule
\textbf{Parameter} & \textbf{Description} & \textbf{Value} \\
\midrule
$N$ & Total training epochs & $100$ \\
$h$ & Training epochs per trajectory & $5$ \\
$\eta_\theta$ & Learning rate of policy network & $5 \times 10^{-5}$ \\
$\eta_\varphi$ & Learning rate of value network & $5 \times 10^{-5}$ \\
$\epsilon$ & PPO clipping value & $0.2$ \\
\bottomrule
\end{tabular}
\label{tab:hyper1}
\end{table}

To avoid redundant retraining in later stages, we jointly train the auxiliary networks $A_\omega$ (advantage predictor) and $D_\psi$ (delay predictor) at this stage. 
The advantage network $A_\omega$ consists of an input BatchNorm layer, three fully connected layers, two SiLU activation layers, and one FiLM layer. The fully connected layers use hidden dimensions of $256$ and $512$, and the output is a single scalar. The FiLM layer itself contains two fully connected layers and one SiLU activation layer, with a hidden dimension of $256$ and an output dimension equal to twice the input dimension.
The delay prediction network $D_\psi$ has a similar architecture, except that it does not include a FiLM layer.
Both auxiliary networks are trained with Adam and a learning rate of $5 \times 10^{-5}$.

\subsubsection{Results and analysis}
We evaluate the trained $\pi_\theta$ on three reasoning benchmarks: \texttt{gsm8k}, \texttt{gaokao2023en}, and \texttt{mmlu\_stem}. 
The accuracies achieved by using only the SLM and only the LLM are taken as the lower and upper bounds, respectively, while the corresponding computation costs are also recorded. 
To further assess our method, we include RSD~\cite{liao2025rewardguided} as a representative baseline. 
For a fair comparison between performance and runtime overhead, we adopt \texttt{Skywork-o1-Open-PRM-Qwen-2.5-1.5B}~\cite{he_2024_16998085} as the PRM used in RSD. The results are shown in Figures~\ref{fig:pi_theta_all}.

We have the following main observations.
\begin{itemize}[leftmargin=0.5cm]
    \item Across all three datasets, $\pi_\theta$ consistently achieves a markedly better trade-off than the all-SLM baseline: it substantially improves accuracy while incurring only a fraction of the computation required by the all-LLM strategy. This confirms that the learned policy can identify which reasoning steps genuinely profit from the stronger LLM and which can be handled efficiently by the local SLM.
    \item Compared with RSD, $\pi_\theta$ is especially advantageous on the more challenging \texttt{gaokao2023en} and \texttt{mmlu\_stem} benchmarks, delivering higher accuracy with lower computational cost. On \texttt{gsm8k} the two methods are highly competitive, exhibiting comparable accuracy-cost trade-offs. These results demonstrate that the proposed screening policy is an effective routing mechanism, not merely a cheap alternative.
    \item Across all datasets, $\pi_\theta$ lies substantially closer to the all-LLM accuracy bound than to the all-SLM bound, while its computation cost remains far below the all-LLM level. This indicates that the policy network successfully captures high-value offloading opportunities without over-using the stronger model.
\end{itemize}

Overall, the results confirm that $\pi_\theta$ provides a strong single-user collaboration policy and a reliable foundation for the full PRADA framework in dynamic multi-user scenarios.

\begin{table*}[t]
\centering
\caption{Generalization performance of PRADA across three reasoning benchmarks.}
\label{tab:performance}
\setlength{\tabcolsep}{4mm}
\begin{tabular}{llccc}
\toprule
\textbf{Dataset} & \textbf{Metric} & \textbf{All SLM} & \boldmath$\pi_\theta$ & \textbf{PRADA} \\
\midrule
\multirow{4}{*}{\texttt{gsm8k}}
& Accuracy (\%) & 85.2 & 93.9 & 90.3 \\
& Avg. Processing Delay (ms/task) & $3.1$ & $9.9$ & $6.0$ \\
& Avg. Communication Delay (ms/task) & -- & -- & $0.4$ \\
& Avg. Queuing Delay (ms/task) & -- & -- & $2.5$ \\
\midrule
\multirow{4}{*}{\texttt{gaokao2023en}}
& Accuracy (\%) & 66.2 & 68.1 & 67.3 \\
& Avg. Processing Delay (ms/task) & $5.8$ & $22.5$ & $11.7$ \\
& Avg. Communication Delay (ms/task) & -- & -- & $0.6$ \\
& Avg. Queuing Delay (ms/task) & -- & -- & $3.2$ \\
\midrule
\multirow{4}{*}{\texttt{mmlu\_stem}}
& Accuracy (\%) & 59.3 & 70.4 & 65.2 \\
& Avg. Processing Delay (ms/task) & $5.6$ & $15.3$ & $8.9$ \\
& Avg. Communication Delay (ms/task) & -- & -- & $0.7$ \\
& Avg. Queuing Delay (ms/task) & -- & -- & $3.8$ \\
\bottomrule
\end{tabular}
\end{table*}

\subsection{Multi-User Evaluation of PRADA}

We now evaluate PRADA in dynamic multi-user scenarios. The system configuration is given in Table~\ref{tab:hyper2}, and unless stated otherwise, this configuration is used throughout the subsequent experiments.

\begin{table}[t]
\caption{Hyperparameter settings for PRADA.}
\label{tab:hyper2}
\centering
\setlength{\tabcolsep}{2mm}
\begin{tabular}{c p{4.6cm} c}
\toprule
\textbf{Parameter} & \textbf{Description} & \textbf{Value} \\
\midrule
$M$ & Maximum parallel capacity of server & $9$ \\
$B$ & Total available bandwidth (bit/s) & $4 \times 10^{7}$ \\
$\lambda$ & User arrival intensity & $3$ \\
$\Delta t$ & Slot length (ms) & $1$ \\
$\mathrm{FlopsLLM}$ & LLM computation capability (FLOPs/s) & $8 \times 10^{13}$ \\
$\mathcal{K}_{\max}$ & Maximum heuristic search iterations & $80$ \\
\bottomrule
\end{tabular}
\end{table}

Table~\ref{tab:performance} summarizes the results on \texttt{gsm8k}, \texttt{gaokao2023en}, and \texttt{mmlu\_stem}. The following conclusions can be drawn.
\begin{itemize}[leftmargin=0.5cm]
    \item PRADA preserves a large portion of the accuracy improvement brought by the policy network $\pi_\theta$ across all three benchmarks. On \texttt{gsm8k}, PRADA achieves 90.3\% accuracy, compared with 93.9\% for $\pi_\theta$ and 85.2\% for the all-SLM baseline. On \texttt{gaokao2023en}, PRADA reaches 67.3\%, remaining close to the 68.1\% achieved by $\pi_\theta$ and still exceeding the all-SLM result of 66.2\%. A similar trend is observed on \texttt{mmlu\_stem}, where PRADA attains 65.2\%, compared with 70.4\% for $\pi_\theta$ and 59.3\% for all-SLM. These results indicate that the second-stage multi-user scheduling mechanism does not substantially undermine the effectiveness of the first-stage routing policy.
    \item Simultaneously, PRADA significantly reduces the average processing delay relative to directly following $\pi_\theta$. On \texttt{gsm8k}, the average processing delay decreases from 9.9 ms/task to 6.0 ms/task. On \texttt{gaokao2023en}, it is reduced from 22.5 ms/task to 11.7 ms/task, which corresponds to nearly a 48\% reduction. On \texttt{mmlu\_stem}, the processing delay drops from 15.3 ms/task to 8.9 ms/task. At the same time, the resulting delay remains only moderately higher than that of the all-SLM baseline. This shows that PRADA successfully translates the strong single-user routing capability of $\pi_\theta$ into a more practical multi-user strategy with substantially lower computation cost.
    \item Communication delay remains very small across all three datasets, ranging from only 0.4 to 0.7 ms/task. This suggests that under the current bandwidth setting, transmission overhead is well controlled and is not the dominant source of latency. By comparison, the queuing delay is larger than the communication delay, reaching 2.5 ms/task on \texttt{gsm8k}, 3.2 ms/task on \texttt{gaokao2023en}, and 3.8 ms/task on \texttt{mmlu\_stem}. Thus, server contention contributes more to additional latency than pure communication overhead, although both remain within a relatively small range.
    \item Our PRADA framework generalizes well across datasets of varying difficulty: in all three cases it maintains a favorable accuracy-latency trade-off, preserving most of the performance gain of $\pi_\theta$ while significantly lowering system cost through resource-aware scheduling.
\end{itemize}

\begin{figure}[t]
    \centering
    \begin{subfigure}[t]{0.48\linewidth}
        \centering
        \includegraphics[width=\linewidth]{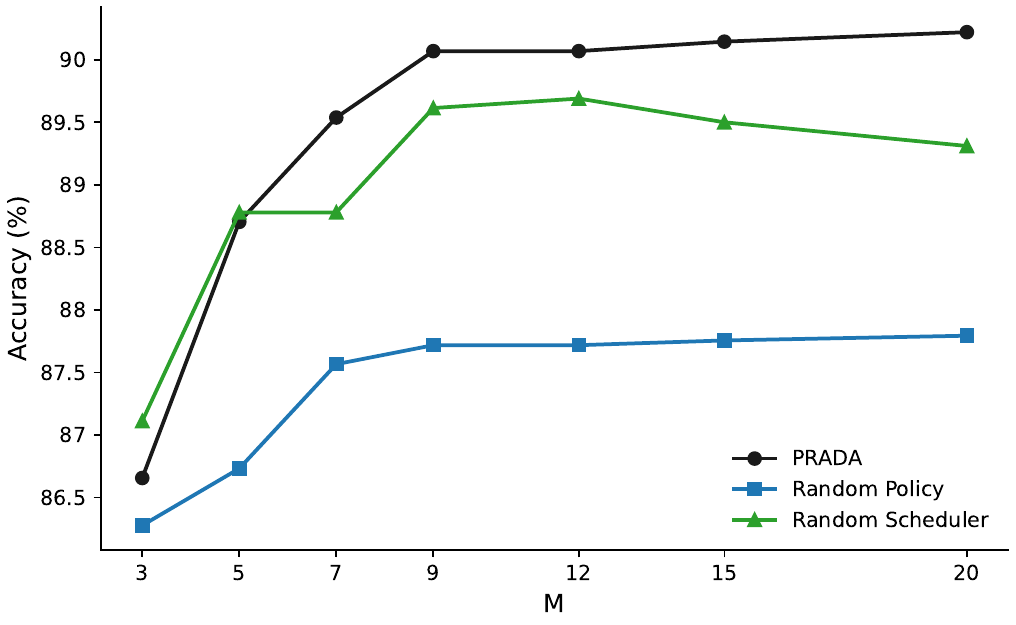}
        \caption{Accuracy vs $M$}
        \label{fig:accuracy_vs_M}
    \end{subfigure}
    \hfill
    \begin{subfigure}[t]{0.48\linewidth}
        \centering
        \includegraphics[width=\linewidth]{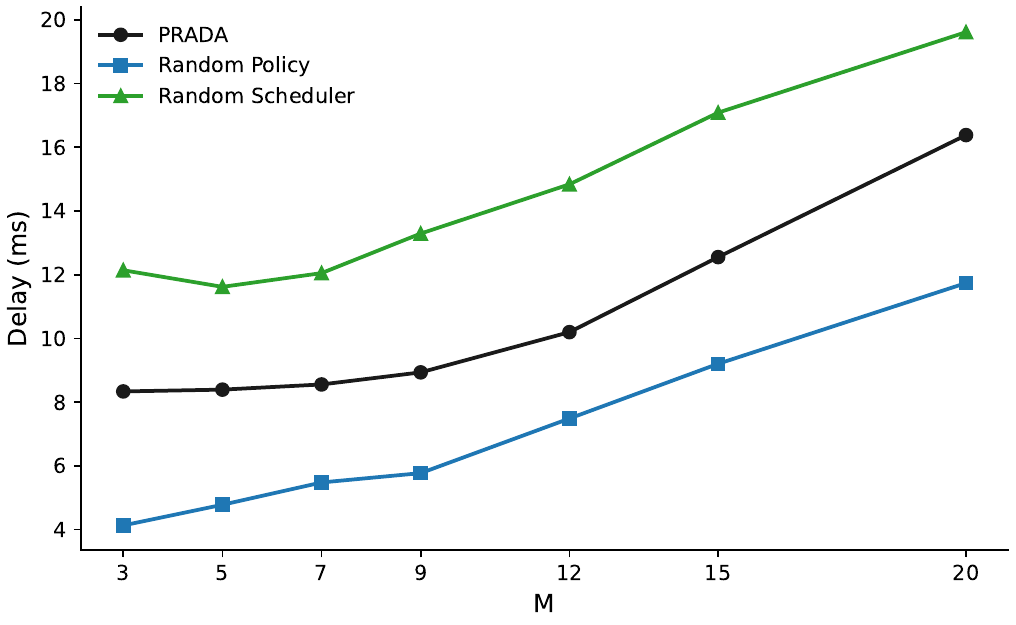}
        \caption{Total Delay vs $M$}
        \label{fig:delay_vs_M}
    \end{subfigure}
    \caption{Performance of PRADA under different server parallel capacities $M$.}
    \label{fig:performance_vs_M}
\end{figure}

\subsection{Impact of the Server Parallel Capacity $M$}

\begin{figure*}[t]
    \centering
    \begin{subfigure}[b]{0.32\textwidth}
        \includegraphics[width=\textwidth]{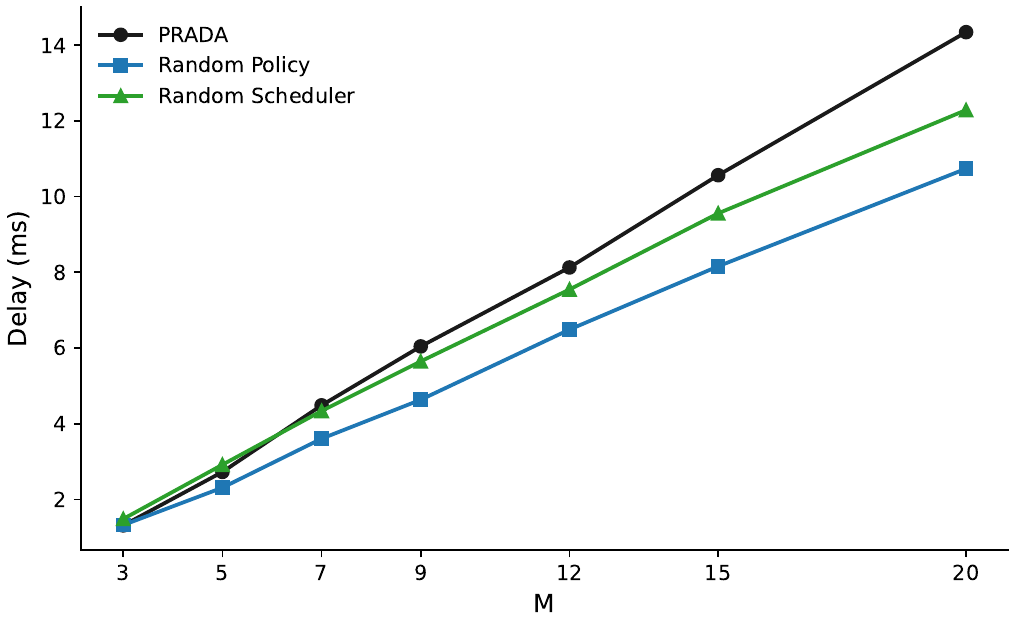}
        \caption{Processing delay versus $M$}
        \label{fig:processing_delay_vs_M}
    \end{subfigure}
    \hfill
    \begin{subfigure}[b]{0.32\textwidth}
        \includegraphics[width=\textwidth]{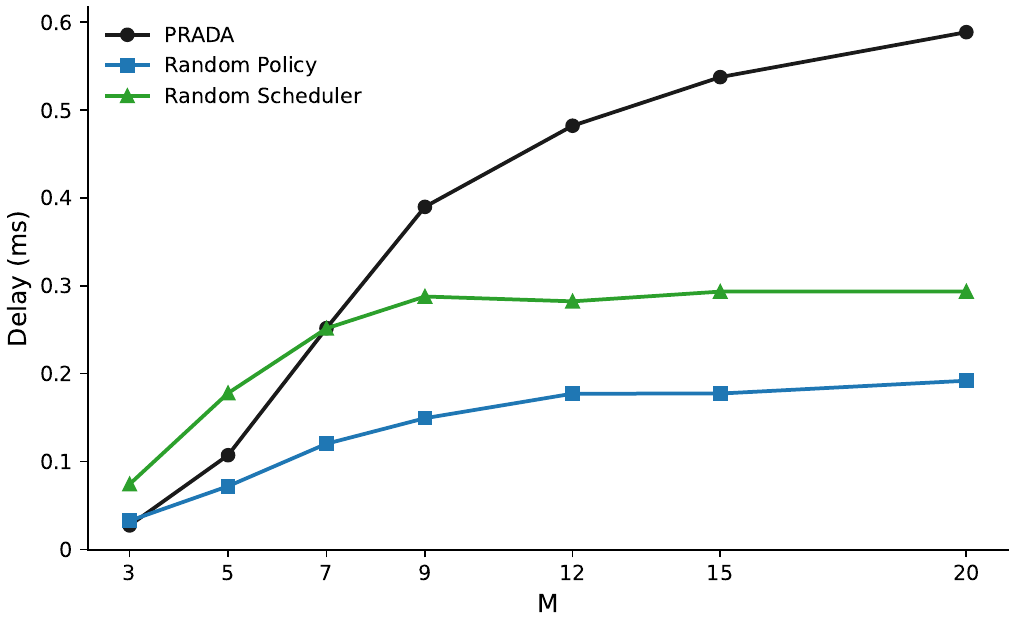}
        \caption{Communication delay versus $M$}
        \label{fig:comm_delay_vs_M}
    \end{subfigure}
    \hfill
    \begin{subfigure}[b]{0.32\textwidth}
        \includegraphics[width=\textwidth]{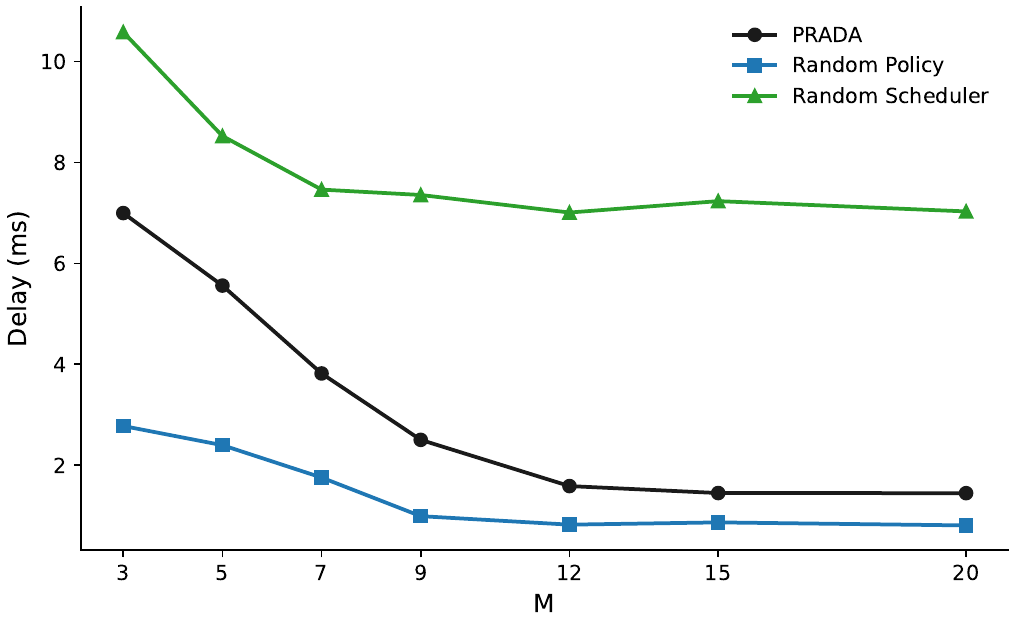}
        \caption{Queuing delay versus $M$}
        \label{fig:queue_delay_vs_M}
    \end{subfigure}
    \caption{Processing, communication, and queuing delays as functions of the server LLM parallel capacity $M$.}
    \label{fig:delay_all}
\end{figure*}

The server parallel capacity $M$ plays a key role in system performance. When $M$ is small, server-worthy requests are delayed by queuing, which increases latency and may degrade task quality. As $M$ increases, this bottleneck is gradually alleviated. However, once $M$ becomes sufficiently large, the marginal gain diminishes, since the system is no longer limited by cloud concurrency. Therefore, selecting an appropriate range of $M$ is essential for balancing accuracy and resource efficiency.

To isolate the contribution of each stage, we introduce two ablation baselines. Random Policy replaces the learned routing decision in Stage~1 with uniform sampling. Random Scheduler preserves the first-stage policy but replaces Stage~2 with a uniform random decision over local fallback, queue admission, and immediate execution. Infeasible sampled actions are mapped to valid alternatives. These baselines are also used in the bandwidth sensitivity study.

Taking the \texttt{gsm8k} dataset as an example, we evaluate the system under different values of $M$. Figure~\ref{fig:performance_vs_M} shows the resulting task accuracy and total delay.

\begin{figure*}[htbp]
    \centering
    \begin{subfigure}[b]{0.32\textwidth}
        \includegraphics[width=\textwidth]{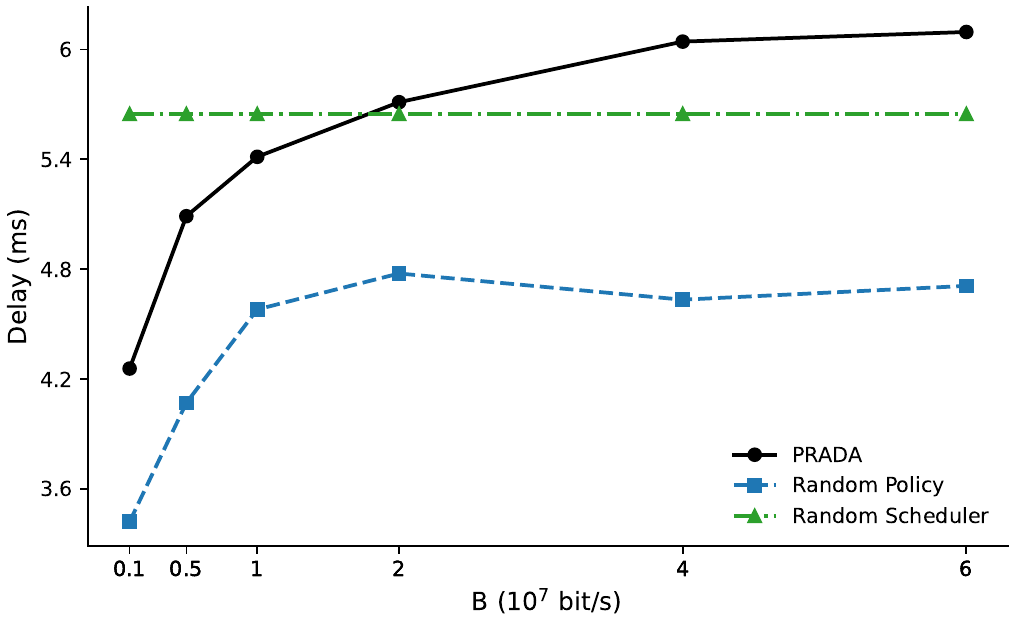}
        \caption{Processing delay versus $B$}
        \label{fig:processing_delay_vs_B}
    \end{subfigure}
    \hfill
    \begin{subfigure}[b]{0.32\textwidth}
        \includegraphics[width=\textwidth]{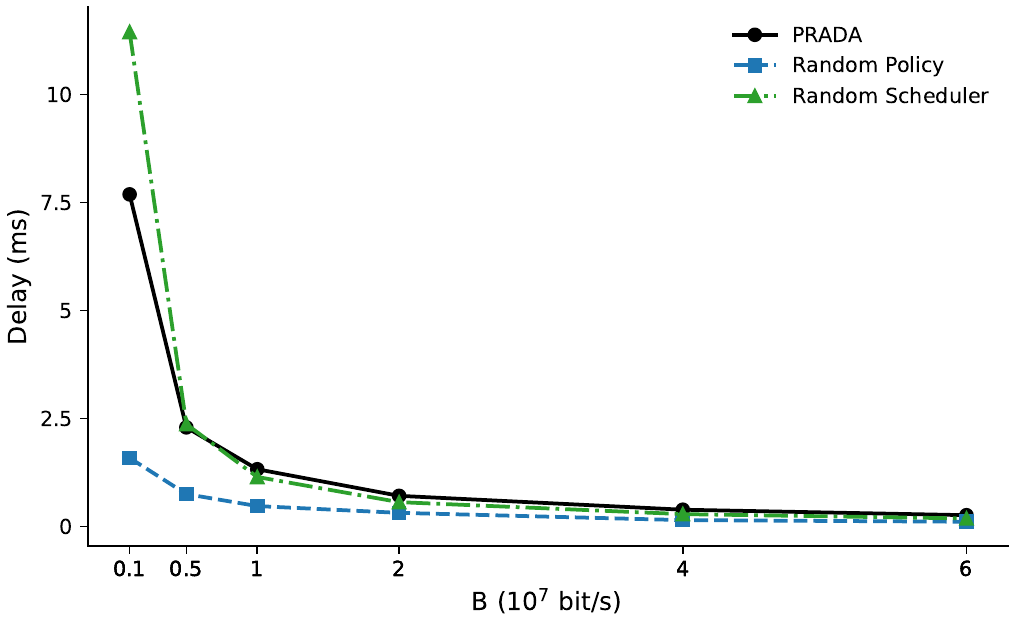}
        \caption{Communication delay versus $B$}
        \label{fig:comm_delay_vs_B}
    \end{subfigure}
    \hfill
    \begin{subfigure}[b]{0.32\textwidth}
        \includegraphics[width=\textwidth]{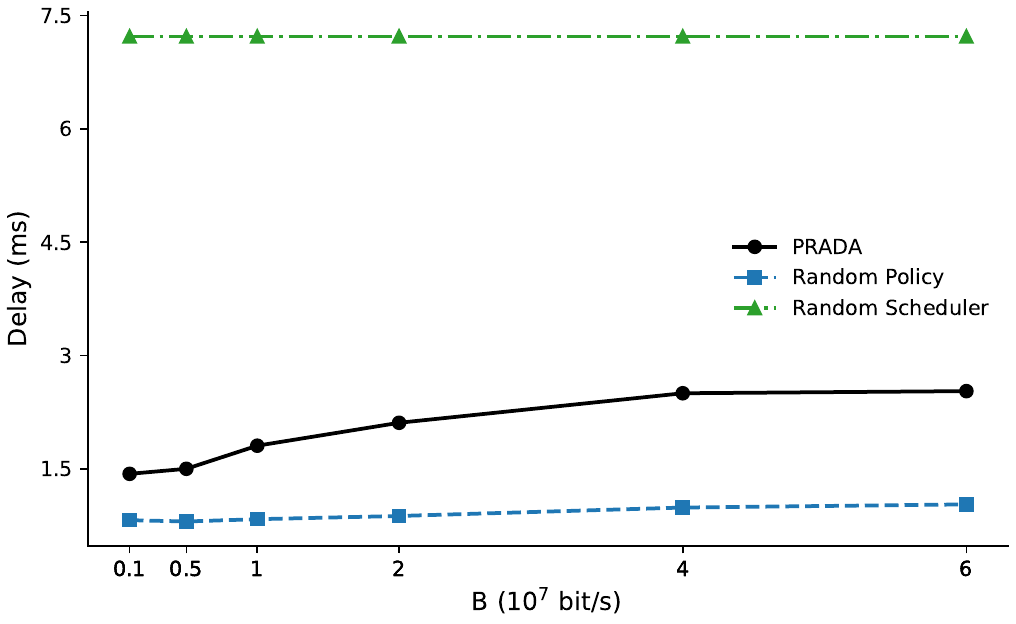}
        \caption{Queuing delay versus $B$}
        \label{fig:queue_delay_vs_B}
    \end{subfigure}
    \caption{Processing, communication, and queuing delays as functions of total bandwidth $B$.}
    \label{fig:delay_all_B}
\end{figure*}

Several observations can be drawn from Figure~\ref{fig:accuracy_vs_M}. As $M$ increases from 3 to 9, PRADA achieves substantial accuracy gains, indicating that the system is initially constrained by limited cloud concurrency. Beyond this range, performance saturates, suggesting that most beneficial offloading decisions have already been realized.

Interestingly, the random scheduler slightly outperforms PRADA in the small-$M$ regime. Under severe resource constraints, only a few requests can be admitted to the cloud, and randomization increases the chance of selecting high-value samples that may be missed by the learned policy. However, this advantage disappears as $M$ grows. Without resource awareness, the random scheduler increasingly allocates capacity to low-value requests, leading to degraded accuracy.

In contrast, PRADA consistently improves with $M$ and remains stable in the high-capacity regime, demonstrating effective utilization of additional resources through coordinated routing and scheduling. This confirms that the performance gain of PRADA does not arise solely from increased capacity, but from the joint design of the two-stage decision framework.

Figure~\ref{fig:delay_all} reports the delay decomposition. The queuing delay of PRADA decreases rapidly with $M$ and approaches zero once sufficient capacity is available, confirming that concurrency is the dominant bottleneck in the low-$M$ regime. Meanwhile, both processing and communication delays increase with $M$, as more requests are offloaded to the cloud, resulting in higher computation and transmission overhead per task.

The ablation baselines exhibit distinct patterns. The random policy suffers from large queuing delay when $M$ is small, while the random scheduler maintains consistently high queuing delay due to its inability to adapt to system congestion. These results highlight the necessity of both structured routing and resource-aware scheduling.

Overall, $M$ exhibits a clear threshold effect: increasing it is highly beneficial when the system is resource-limited, but provides limited gain once queuing is eliminated. This suggests that moderate provisioning is sufficient to achieve most of the attainable performance improvement.

\subsection{Impact of Total Bandwidth $B$}

The total bandwidth $B$ determines how efficiently edge-side SLMs can transmit contexts to the cloud LLM. When $B$ is small, many server-worthy requests cannot be offloaded in time, limiting the benefit of collaboration. Increasing $B$ alleviates this communication bottleneck, while the marginal gain diminishes once bandwidth is no longer the dominant constraint.

Using the \texttt{gsm8k} dataset with $M=9$, we evaluate PRADA under different bandwidth values. Figure~\ref{fig:accuracy_vs_B} shows the resulting task accuracy. As $B$ increases from the low-bandwidth regime, the accuracy of PRADA improves steadily, with a more pronounced gain around $10^6$ to $10^{7}$, and then saturates. This indicates a clear threshold effect: once bandwidth is sufficient to support offloading of longer-context requests, the benefit of LLM collaboration is largely realized.

\begin{figure}[t]
    \centering
    \begin{subfigure}[t]{0.48\linewidth}
        \centering
        \includegraphics[width=\linewidth]{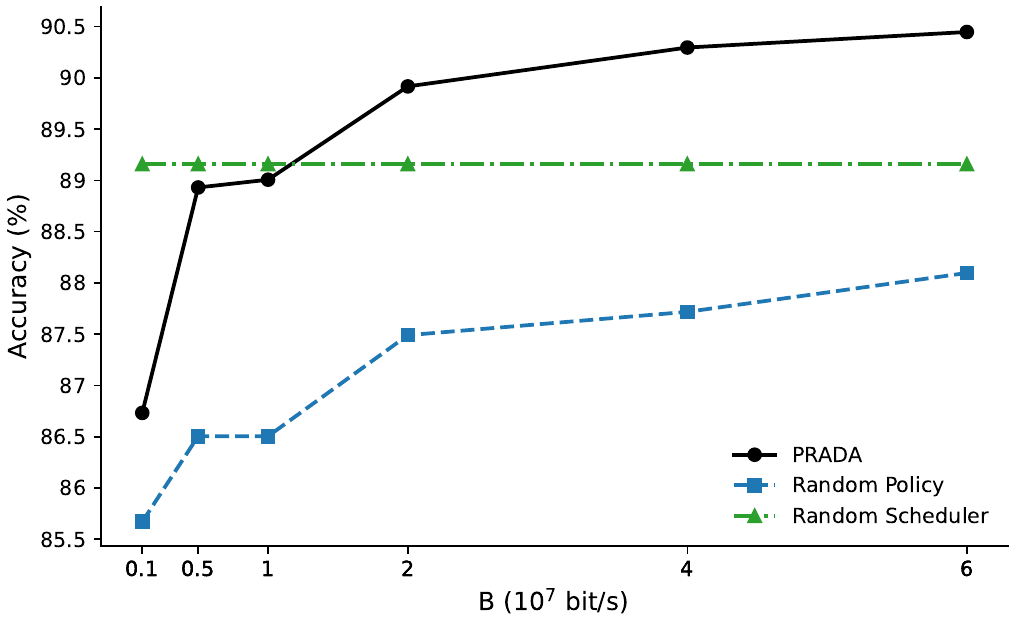}
        \caption{Accuracy vs $B$}
        \label{fig:accuracy_vs_B}
    \end{subfigure}
    \hfill
    \begin{subfigure}[t]{0.48\linewidth}
        \centering
        \includegraphics[width=\linewidth]{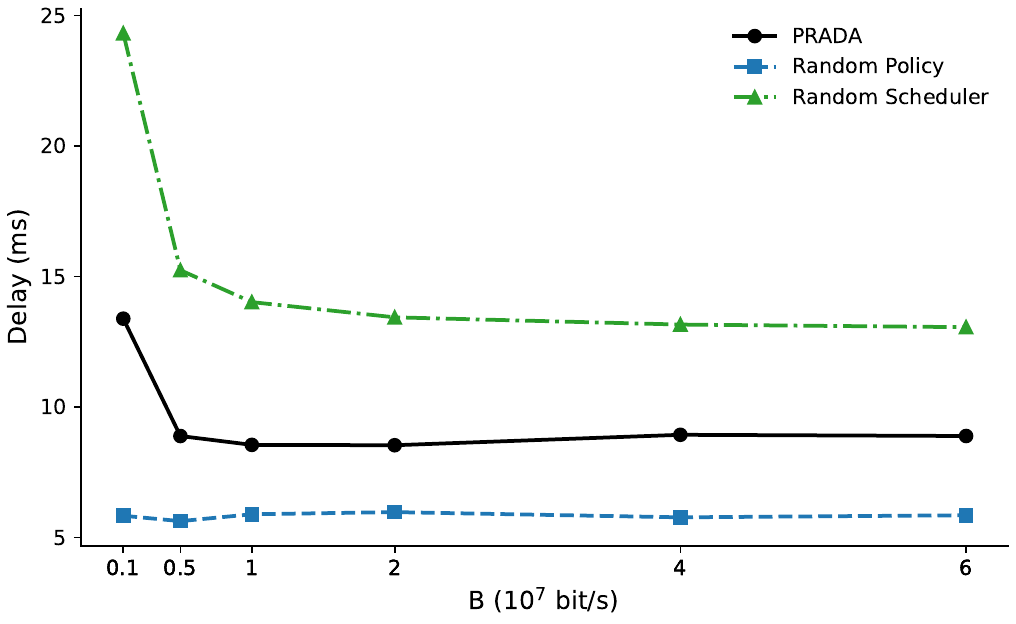}
        \caption{Total Delay vs $B$}
        \label{fig:delay_vs_B}
    \end{subfigure}
    \caption{Performance of PRADA under different bandwidth $B$.}
    \label{fig:performance_vs_B}
\end{figure}

The ablation baselines reveal a distinct behavior in the extremely low-bandwidth regime. The random scheduler achieves relatively high accuracy when $B$ is very small. However, this comes at the cost of dramatically increased latency(Fig.~\ref{fig:delay_vs_B}). This is because it continues to offload requests without accounting for bandwidth scarcity, allowing some difficult tasks to benefit from LLM processing, but incurring excessive transmission and waiting overhead. Therefore, its higher accuracy does not translate to better overall system performance. In contrast, PRADA adopts a more conservative strategy under severe bandwidth constraints, achieving a more balanced accuracy-delay trade-off.

Figure~\ref{fig:delay_all_B} further shows the delay decomposition. Communication delay decreases rapidly with $B$ and becomes negligible in the moderate regime. In contrast, the processing and queuing delays of PRADA increase from the very low-bandwidth regime and then stabilize. This is because higher bandwidth enables more requests to be offloaded, increasing both server-side computation and contention for limited parallel capacity. As a result, the system bottleneck shifts from communication to computation and queuing as $B$ increases.

Overall, $B$ exhibits a step-like transition behavior. When bandwidth is below a critical point, communication dominates system performance and severely limits effective offloading. Once $B$ exceeds this critical threshold, the communication constraint becomes inactive, and further increases provide only marginal benefit. This suggests that bandwidth provisioning should be coordinated with server parallel capacity rather than scaled independently.

\section{Conclusion}
\label{sec:conclusion}
This work addressed the fundamental tension between reasoning quality and latency in dynamic edge networks where heterogeneous agents must collaborate under stochastic arrivals and coupled resource constraints.
The key insight from our PRADA framework is that the PRM, ordinarily a prohibitive online burden, can be confined entirely to offline training, distilling its long-horizon quality assessment into a compact network that incurs negligible runtime cost yet preserves the ability to identify steps that genuinely profit from stronger reasoning.

Several structural lessons emerge. 
First, by proving that a first-stage decision to remain local retains its optimality even after communication and queuing costs are reintroduced, we show that the edge screening and server scheduling stages can be designed and trusted independently, greatly simplifying system architecture. 
Second, the simulation results reveal sharp threshold phenomena: both server parallel capacity and total bandwidth exhibit critical values beyond which quality and latency gains quickly saturate, and the dominant bottleneck shifts from queuing to computation or from communication to contention. These saturation points provide concrete guidance for cost-effective provisioning: one should jointly scale computation and communication to a moderate level rather than over-provisioning either dimension in isolation.

The implications of PRADA extend beyond the specific setting studied. The idea of using capabilities-rich but computationally expensive teacher models exclusively offline to supervise lightweight online orchestrators is broadly applicable to other resource-constrained multi-agent systems.

\appendices

\section{Proof of Proposition \ref{thm:1}}
\label{sec:AppA}
For a local decision ($a_{i,t}=0$), the two value functions coincide because no communication or queuing is incurred:
\begin{equation}
\widetilde{Q}_{i,t}(0) = Q_{\pi_\theta}(x_{i,t_{K_i}}^{(K_i)}, 0).
\end{equation}
For an offloading decision ($a_{i,t}=1$), the true reward subtracts additional latency penalties that are absent in the coarse reward:
\begin{equation}
\widetilde{Q}_{i,t}(1) = Q_{\pi_\theta}(x_{i,t_{K_i}}^{(K_i)}, 1) - \beta(\mathrm{TC}_{i,t} + \mathrm{TQ}_{i,t} ),
\end{equation}
where the expectation is taken over channel conditions and queue states. Since $\beta \ge 0$ and $\mathrm{TC}_{i,t}, \mathrm{TQ}_{i,t} \ge 0$, we have
\begin{equation}
\widetilde{Q}_{i,t}(1) \le Q_{\pi_\theta}(x_{i,t_{K_i}}^{(K_i)}, 1).
\end{equation}

Assume $Q_{\pi_\theta}(x_{i,t_{K_i}}^{(K_i)}, 1) \le Q_{\pi_\theta}(x_{i,t_{K_i}}^{(K_i)}, 0)$. Then
\begin{align*}
\widetilde{Q}_{i,t}(1)
\le Q_{\pi_\theta}(x_{i,t_{K_i}}^{(K_i)}, 1) \
\le Q_{\pi_\theta}(x_{i,t_{K_i}}^{(K_i)}, 0) \
= \widetilde{Q}_{i,t}(0).
\end{align*}
Thus $\widetilde{Q}_{i,t}(1) \le \widetilde{Q}_{i,t}(0)$, so the local action is at least as good as offloading under the true objective as well.

\section{Proof of Theorem~\ref{thm:stage2_policy}}
\label{sec:AppB}

For fixed multipliers $\lambda_s, \mu$, the Lagrangian separates into per‑request terms. The optimal action for request $i$ is the one that maximizes
\[
\widetilde{Q}_{i,t}(\alpha_{i,t}) - \lambda_s \mathbb{I}(\alpha_{i,t}=2) - \mu B_{i,t}\mathbb{I}(\alpha_{i,t}=2).
\]
Using the definition of $\widetilde{Q}_{i,t}(\alpha_{i,t})$ from \eqref{eq:modified_q_stage2}, we compare the three candidates $\alpha_{i,t}=0,1,2$.

\textit{Case I: Competitive admission ($|\mathcal{O}_t| > M - |\mathcal{M}_t|$).} Here, multiple requests compete for limited server slots. Immediate admission ($\alpha_{i,t}=2$) is optimal if:
\begin{equation}
\left\{
\begin{aligned}
\widetilde{Q}_{i,t}(2)-\lambda_s-\mu B_{i,t} &\ge \widetilde{Q}_{i,t}(1),\\
\widetilde{Q}_{i,t}(2)-\lambda_s-\mu B_{i,t} &\ge \widetilde{Q}_{i,t}(0).
\end{aligned}
\right.
\end{equation}
By combining Eq.~\eqref{eq:modified_q_stage2} and simplifying, we obtain:
\begin{equation}
\left\{
\begin{aligned}
&\beta\,\mathrm{TQ}_{i,t} - \mu B_{i,t}\ge \lambda_s,\\
&Q_{\pi_\theta}\bigl(x_{i,t}^{(k_{i,t})},1\bigr)
  - Q_{\pi_\theta}\bigl(x_{i,t}^{(k_{i,t})},0\bigr) -\beta\,\mathrm{TC}_{i,t}-\mu B_{i,t}\ge \lambda_s.
\end{aligned}
\right.
\end{equation}
Substituting $w_{i,t}$ and $\overline{w}_{i,t}$ into the conditions, we obtain:
\begin{equation}
\left\{
\begin{aligned}
&\overline{w}_{i,t}\ge \lambda_s,\\
&w_{i,t}\ge \lambda_s.
\end{aligned}
\right.
\end{equation}

By analogy, queue admission ($\alpha_{i,t}=1$) is optimal if:
\begin{equation}
\left\{
\begin{aligned}
\widetilde{Q}_{i,t}(1) &> \widetilde{Q}_{i,t}(2)-\lambda_s-\mu B_{i,t},\\
\widetilde{Q}_{i,t}(1) &\ge \widetilde{Q}_{i,t}(0),
\end{aligned}
\right.
\end{equation}
and local return ($\alpha_{i,t}=0$) is optimal if:
\begin{equation}
\left\{
\begin{aligned}
\widetilde{Q}_{i,t}(0) &> \widetilde{Q}_{i,t}(2)-\lambda_s-\mu B_{i,t},\\
\widetilde{Q}_{i,t}(0) &> \widetilde{Q}_{i,t}(1).
\end{aligned}
\right.
\end{equation}
By applying the same simplification to the remaining two cases, we obtain:
\begin{equation}
\alpha_{i,t}^* =
\begin{cases}
1, & \overline{w}_{i,t} < \lambda_s \text{ and } w_{i,t} \ge \overline{w}_{i,t}, \\[1mm]
0, & w_{i,t}<\lambda_s \text{ and }w_{i,t} < \overline{w}_{i,t}.
\end{cases}
\end{equation}
In conclusion, we gain the optimal action structure:
\begin{equation}
\alpha_{i,t}^* =
\begin{cases}
2, & \min(\overline{w}_{i,t}, w_{i,t}) \ge \lambda_s, \\[1mm]
1, & \overline{w}_{i,t} < \lambda_s \text{ and } w_{i,t} \ge \overline{w}_{i,t}, \\[1mm]
0, & w_{i,t}<\lambda_s \text{ and }w_{i,t} < \overline{w}_{i,t}.
\end{cases}
\label{eq:case1_rule}
\end{equation}

\textit{Case II: Fully occupied server ($|\mathcal{M}_t| = M$).} When no immediate slots are available, the decision reduces to a binary choice between queuing and local fallback. Queue admission ($\alpha_{i,t}=1$) is optimal if:
\[
\widetilde{Q}_{i,t}(1) \ge \widetilde{Q}_{i,t}(0)\implies w_{i,t}\ge \overline{w}_{i,t}.
\]

The threshold structure remains: 
\begin{equation}
\alpha_t^{(i)*}
=
\begin{cases}
1, & w_{i,t}\ge \overline{w}_{i,t},\\[1mm]
0, & \text{otherwise}.
\end{cases}
\label{eq:case2_rule}
\end{equation}

\textit{Case III: No admission competition ($|\mathcal{O}_t| \le M - |\mathcal{M}_t|$).} 
When server capacity suffices for all candidates, each request is either immediately admitted or returned locally. 
Immediate admission ($\alpha_{i,t}=2$) is optimal if
\[
\widetilde{Q}_{i,t}(2) - \lambda_s - \mu B_{i,t} \ge \widetilde{Q}_{i,t}(0) 
\implies w_{i,t} \ge \lambda_s,
\]
Since $|\mathcal{O}_t| \le M-|\mathcal{M}_t|$, the server concurrency constraint is non-binding in this case. By the complementary slackness condition of the KKT system,
\[
\lambda_s\!\left(\sum_{i\in\mathcal{O}_t}\mathbb{I}(\alpha_{i,t}=2)-M+|\mathcal{M}_t|\right)=0,
\]
the associated dual variable must satisfy $\lambda_s=0$. Therefore, the immediate-admission condition reduces to $w_{i,t}\ge 0$.
Then, we can obtain the optimal threshold structure:
\begin{equation}
\alpha_{i,t}^* =
\begin{cases}
2, & w_{i,t} \ge 0, \\[1mm]
0, & \text{otherwise},
\end{cases}
\label{eq:case3_rule_new}
\end{equation}

Combining these three cases, we conclude that the optimal Stage~2 action for each request exhibits a threshold-based structure with respect to the effective offloading gain, as stated in the theorem.

\bibliographystyle{IEEEtran}
\bibliography{References}

\end{document}